\documentclass[journal]{IEEEtran}

\newcommand{\TITLE}{Leader-Contention-Based User Matching for 802.11 Multiuser MIMO Networks}
\newcommand{\MYKEYWORDS}{Multiuser MIMO, User matching, Channel
orthogonality}

\newcommand{\name}{\textsf{MIMOMate}}
\newcommand{\xref}[1]{Section \ref{#1}}

\usepackage{color}
\usepackage{colortbl}
\newcommand{\textred}[1]{\textcolor{red}{#1}}

\ifx\noeditingmarks\undefined
	\newcommand{\pgwrapper}[2]{\textred{#1: #2}}
\else
   \newcommand{\pgwrapper}[2]{}
\fi

\usepackage{cite}      

\usepackage{graphicx,epsfig}  

\usepackage{subfigure}
\usepackage{amsfonts}
\usepackage{amsmath}   
\usepackage{amsthm}
\usepackage{amssymb}
\usepackage{verbatim} 
\def\imagetop#1{\vtop{\null\hbox{\centerline{#1}}\vskip -0.15in}}
\usepackage{array}

\theoremstyle{definition}

\newtheorem{probdefi}{Problem}
\theoremstyle{Theorem}
\newtheorem{thrm}{Theorem}

\usepackage{setspace}
\usepackage{multirow} 
\usepackage[linesnumbered,ruled,vlined]{algorithm2e}
\usepackage{listings}         
\usepackage{mathrsfs} 

\graphicspath{{./}}
\DeclareGraphicsExtensions{.jpg,.eps,.pnm}
\begin{document}
%
\title{\TITLE}
%
%
\author{\IEEEauthorblockN{$^{\dag\ddag}$Tung-Wei Kuo, $^{\dag}$Kuang-Che
	Lee, $^{\dag}$Kate Ching-Ju Lin and
$^{\ddag}$Ming-Jer Tsai} \\
\IEEEauthorblockA{$^{\dag}$Research Center for Information Technology
Innovation, Academia Sinica, Taiwan\\
$^{\ddag}$Department of Computer
Science, National Tsing Hua University, Taiwan}
}
\bibliographystyle{IEEEtran}
\maketitle
\thispagestyle{plain}
\begin{abstract} 
In multiuser MIMO (MU-MIMO) LANs, the achievable throughput of a
client depends on who are transmitting concurrently with it. Existing
MU-MIMO MAC protocols however enable clients to use the traditional
802.11 contention to contend for concurrent transmission opportunities
on the uplink. Such a contention-based protocol not only wastes lots
of channel time on multiple rounds of contention, but also fails to
maximally deliver the gain of MU-MIMO because users randomly join
concurrent transmissions without considering their channel
characteristics. To address such inefficiency, this paper introduces
\name, a leader-contention-based MU-MIMO MAC protocol that matches
clients as concurrent transmitters according to their channel
characteristics to maximally deliver the MU-MIMO gain, while ensuring
all users to fairly share concurrent transmission opportunities.
Furthermore, \name\ elects the leader of the matched users to contend
for transmission opportunities using traditional 802.11 CSMA/CA. It
hence requires only a single contention overhead for concurrent
streams, and can be compatible with legacy 802.11 devices.  A
prototype implementation in USRP-N200 shows that \name\ achieves an
average throughput gain of 1.42x and 1.52x over the traditional
contention-based protocol for 2-antenna and 3-antenna AP scenarios,
respectively, and also provides fairness for clients.
\end{abstract}

\section{Introduction}\label{sec:intro}
With the growing technique of multiple antenna systems, the number of
antennas on an access point (AP) is increasing steadily. Most of
mobile devices, such as smartphones or tablets, are however limited by
their size and power constraints, and hence have a fewer number of
antennas as compared to the AP.  Traditional 802.11 protocols, which
enable only a single client to communicate with the AP, hence cannot
fully utilize concurrent transmission opportunities supported by a
multi-antenna AP.  To address this problem, recent work has advocated
developing multiuser MIMO (MU-MIMO) LANs~\cite{beamforming,
beamforming2, megamimo,SAM, TurboRate,JBM,Argos}  to enable multiple
clients to communicate concurrently with an AP and fully utilize all
the available {\em degrees of freedom}~\cite{TseV:05}.

\begin{figure}[t!]
\centering
{
\footnotesize
\begin{tabular}
	{>{\centering\arraybackslash}m{1.85in}>{\centering\arraybackslash}m{1.5in}}
\imagetop{\epsfig{file=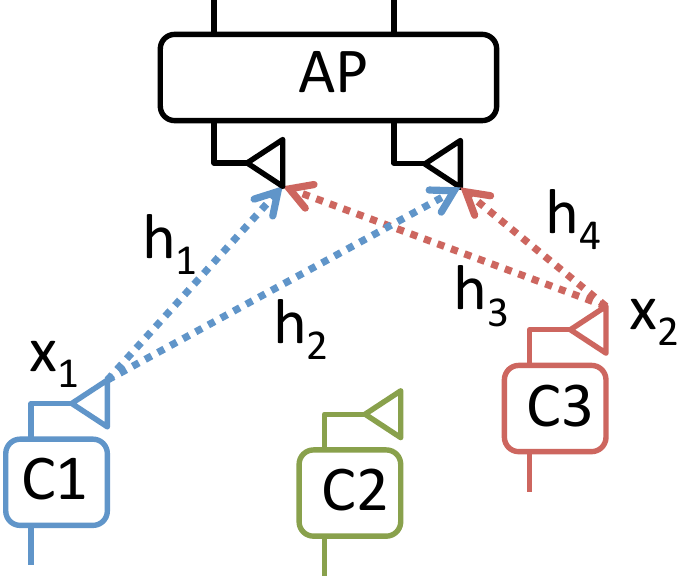, scale=.45}}
& 
\imagetop{\epsfig{file=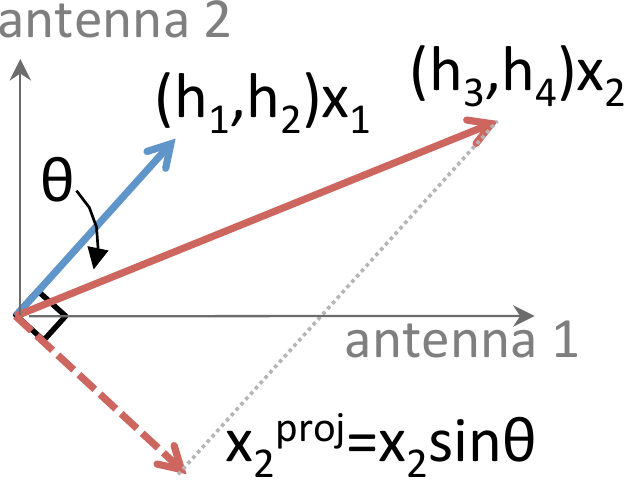, scale=.45}}
\\
(a) client C3 sends concurrently with C1
& 
(b) SNR after projection changes\\
\imagetop{\epsfig{file=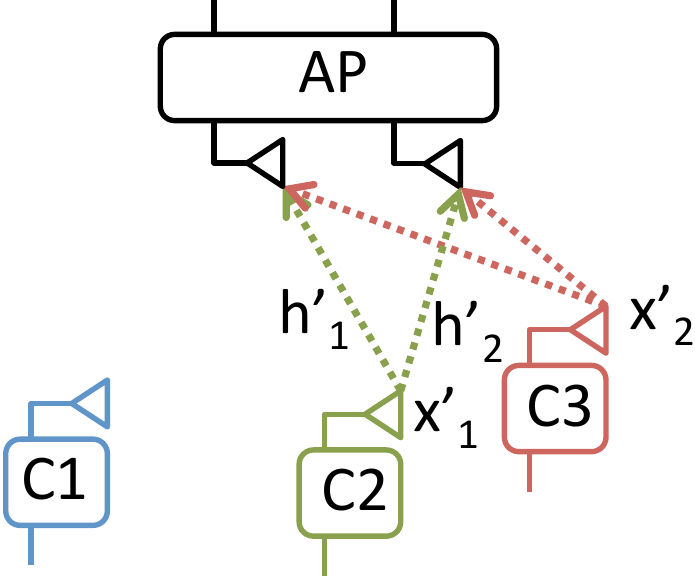, scale=.45}}
&
\imagetop{\epsfig{file=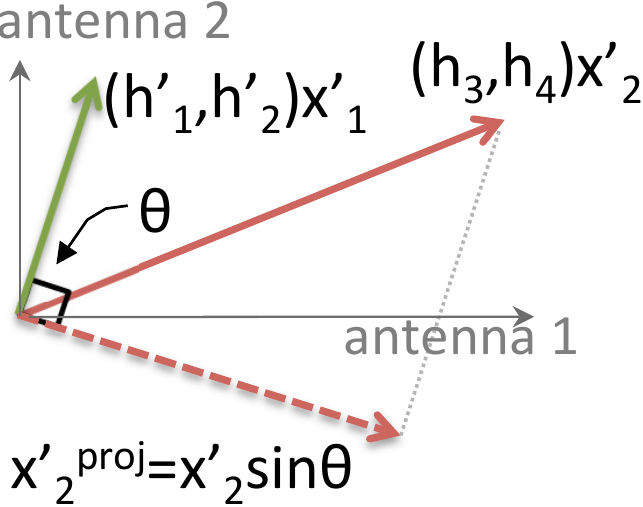, scale=.45}}
\\
(c) client C3 sends concurrently with C2
& 
(d) SNR after projection is larger than in (b)
\end{tabular}
}
\vspace{-12pt}
\caption{SNR after projection changes}
\label{fig:example}
\end{figure}

Though some MU-MIMO MAC protocols~\cite{beamforming, beamforming2,
megamimo,SAM, TurboRate, JBM, Argos} have been proposed to realize
concurrent transmissions across different nodes, in either the uplink
or downlink scenarios, they simply select a random subset of users to
communicate concurrently with an AP.  However, in a MU-MIMO LAN, the
achievable throughput of a client highly depends on who are
transmitting concurrently with it.  Consider the example in
Fig.~\ref{fig:example}(a), where three single-antenna clients contend
for communicating with a 2-antenna AP.  Say clients C1 and C3 win the
first and second contentions, respectively, and transmit concurrently
to the AP. The 2-antenna AP receives the signals in a 2-dimensional
antenna space, as shown in Fig.~\ref{fig:example}(b). The basic
approach for decoding the concurrent packets is called {\em
zero-forcing with successive interference cancellation}
(ZF-SIC)~\cite{TurboRate}\cite{TseV:05}. The AP first zero-forces (ZF)
the interference from client C1 by projecting the signal of the second
client C3 on a direction orthogonal to C1, and hence can decode C3.
The AP then uses interference cancellation (IC) to subtract C3's
signal and decode the first client C1.  We note that the second client
decoded by ZF however experiences SNR reduction after projection, as
shown in Fig.~\ref{fig:example}(b), while the first client decoded by
IC obtains the same SNR as it transmits alone.  The amount of SNR
reduction for the second client depends on channel orthogonality
between two concurrent clients, i.e., the angle $\theta$ between C1
and C3.  Consider another example in Fig.~\ref{fig:example}(c), where
clients C2 and C3 transmit the first and second streams, respectively.
Client C3 in this case gets a higher SNR after projection, as in
Fig.~\ref{fig:example}(d), and thus achieves a higher throughput, as
compared to sending concurrently with C1, because its channel is more
orthogonal to C2's channel. This example illustrates that random user
selection cannot efficiently deliver the gain of a MU-MIMO LAN.

To better deliver MU-MIMO gains, some theoretical works on {\em
downlink} MU-MIMO LANs have focused on allowing the single
transmitter, i.e., AP, to select a proper subset of users as
concurrent receivers. User selection in the {\em uplink} scenario is
however much more challenging because multiple contending clients
compete for transmission opportunities without coordination. As a
result, existing uplink MU-MIMO LANs~\cite{SAM}\cite{TurboRate} still
adopt traditional 802.11 contention to ask all the clients to contend
for concurrent transmissions without sophisticated user selection.
Such a contention-based scheme not only fails to select concurrent
transmitters according to their channel characteristics, but also
wastes a lot of channel time for multiple rounds of contention for
concurrent transmission opportunities.  For example, if multiple
single-antenna clients contend for transmitting to a 3-antenna AP,
they need to contend for three transmission opportunities
sequentially. If a client fails to win the first contention, it needs
another round of contention to contend for the opportunity of sending
the second stream. If it fails again, it needs to contend for
transmitting the third stream. The channel time occupied by such
sequential contention could significantly offset the gain of MU-MIMO
LANs.

Though some prior studies work on uplink MU-MIMO user scheduling, they
either maximize the sum rate of concurrent transmissions without
considering fairness~\cite{uplink-1}\cite{uplink-2}, or cannot be
compatible with legacy 802.11 contention-based MAC~\cite{multiround,
multiround2, ITC-Jung-12, uplink-6,uplink-7,uplink-3,uplink-4}. In
addition, all the above proposals do not give any formal model to
formulate the relationship between two conflict design goals, i.e.,
throughput and fairness. Hence, in this paper, we propose \name, a
{\em leader-contention-based} MU-MIMO MAC protocol that matches
concurrent transmitters according to their channel characteristics to
maximally deliver the MU-MIMO gain under the fairness requirement. Our
contributions are as follows:
\begin{itemize}
\item We first formulate, in Section~\ref{sec:match}, a rigorous model
	to formally define the user selection problem of maximizing
	throughput under the fairness constraint, and propose a matching
	algorithm to solve the problem.  We further show that \name's user
	matching algorithm can achieve the optimal solution for the
	2-antenna AP scenario.
\item We propose a MU-MIMO MAC design, in Section~\ref{sec:mac}, that
	integrates our proposed user matching algorithm with a {\em
	leader-based} contention scheme.  \name's leader-contention-based
	MAC is hence compatible with traditional 802.11, and, more
	importantly, requires only a single contention overhead not
	scaling up with the number of concurrent streams.
\item Unlike prior theoretical work that only mathematically analyzes
	the effect of MU-MIMO user selection, we build a prototype of
	\name\ using the USRP-N200 radio platform~\cite{usrp}, and use
	testbed measurements to understand the inefficiency of random user
	selection in real channels in Section~\ref{sec:measurement}.
\item We finally experimentally evaluate, in Section~\ref{sec:result},
	the performance of \name. The results show that \name\ achieves an
	average throughput gain of 1.42x and 1.52x over the
	sequential-contention-based protocol for 2-antenna and 3-antenna
	AP scenarios, respectively, and also provides users fair
	transmission opportunities.

\end{itemize}

The remainder of this paper is organized as follows.  We review
related works in Section~\ref{sec:related}.
Section~\ref{sec:measurement} measures how existing schemes fail to
deliver MU-MIMO gains and provide fairness in real channels.
Sections~\ref{sec:match} and~\ref{sec:mac} describe our \name\
algorithm and how to realize it as a MAC protocol, respectively.  In
Sections~\ref{sec:result} and~\ref{sec:simulation}, we evaluate the
performance of \name\ via experiments and simulations, respectively.
Finally, Section~\ref{sec:conclusion} concludes this paper.
\section{Related Work}\label{sec:related}
In the last few years, the advantage of MU-MIMO LANs has been verified
theoretically~\cite{theory-mumimo,theory-mumimo-2,theory-mumimo-3} and
demonstrated
empirically~\cite{beamforming,beamforming2,megamimo,SAM,TurboRate,IAC,n+}.
In Beamforming~\cite{beamforming,beamforming2,megamimo}, a
multi-antenna AP uses the precoding technique to transmit multiple
streams to multiple single-antenna clients.  SAM~\cite{SAM} focuses on
the uplink scenario and allows multiple single-antenna clients to
communicate concurrently with a multi-antenna AP.
TurboRate~\cite{TurboRate} proposes a rate adaptation protocol for
uplink MU-MIMO LANs.  IAC~\cite{IAC} connects multiple APs through the
Ethernet to form a virtual MIMO node that communicates concurrently
with multiple clients.  802.11n$^+$~\cite{n+} enables concurrent
transmissions across different links. All the above practical MU-MIMO
systems leverage the traditional 802.11 content mechanism to share
concurrent transmission opportunities. In contrast, \name\ enables
clients with a better channel orthogonality to transmit concurrently.

Prior theoretical work on user selection in downlink MU-MIMO
LANs~\cite{downlink-1,downlink-2,downlink-3,downlink-4,downlink-5,CCNC-Magd-13}
selects the optimal subset of clients from those who have packets
queued in the AP to maximize the sum rate of concurrent transmissions.
The works \cite{downlink-3}\cite{downlink-5} further address the issue
of fairness. However, their solutions are designed for downlink
MU-MIMO, and cannot be easily applied in the uplink scenarios due to
the lack of a coordinator.

User selection in uplink MU-MIMO LANs~\cite{uplink-1}\cite{uplink-2}
requires the AP to explicitly coordinate between the clients for every
packet and select the optimal subset of clients to transmit at the
rate specified by the AP. Enabling coordination among clients for
uplink traffic however requires a significant signaling overhead.  Our
work differs from those user selection algorithms in that it matches
multiple potential subsets of concurrent transmitters to improve the
system throughput, but elects a leader from the matched users to
perform traditional 802.11 contention without coordination among
clients. Some previous works
\cite{SAM,multiround,multiround2,ITC-Jung-12} propose to use
multi-round contention to enable as many concurrent transmissions as
possible, while giving clients a fair opportunity to transmit
concurrent streams. SAM~\cite{SAM} proposes a preamble-counting
protocol, which allows each client to count the number of existing
streams and determine whether it can contend for sending a concurrent
stream in a distributed way. Multi-round
contention~\cite{multiround}\cite{multiround2} is proposed to let
clients send RTSs in multiple rounds of contention. Multiple clients
might send RTSs concurrently in each single round of contention, and
the AP then feedbacks a CTS to notify those clients that can transmit
concurrently. An asynchronous MAC protocol \cite{ITC-Jung-12} is
proposed to enable clients to independently start their concurrent
transmissions, i.e., without the need of starting concurrent
transmissions at the same time. It however relies on a control channel
to feedback who can join concurrent transmissions.  There are also
some papers that have considered channel orthogonality and fairness
jointly in uplink MU-MIMO~\cite{uplink-3,uplink-4,uplink-6,uplink-7}.
However, these solutions are heuristics without any formal performance
analysis, and also not compatible with the existing 802.11 standard.
Our work is the first that maximizes the throughput under the fairness
constraint, and is able to coexist with legacy 802.11 devices.
\section{MU-MIMO Background and Motivations} \label{sec:measurement}
Before describing our proposed protocol, we first use testbed
measurements in real channels to understand the limitation of the
existing MU-MIMO MAC protocols.  The measurement results also give us
an insight to the motivation of enabling user selection in a MU-MIMO
MAC.  We consider again the network in Fig.~\ref{fig:example}(a) where
two single-antenna clients communicate with a 2-antenna AP.  The
measurements are empirically performed using USRP-N200~\cite{usrp} on
a 10 MHz OFDM channel with 802.11 modulations and coding rates.  The
available bit-rates hence range from 3--27 Mb/s. The measurements are
designed to answer the following questions: 

\vskip 0.05in \noindent {\bf a) How often is a client unable to
transmit concurrently?} 
Recall that the second client's SNR reduces after projection, and the
amount of SNR reduction depends on the angle between its channel and
the channel of its concurrent transmitter, i.e., the first client.
This means that the throughput of the second client in a MU-MIMO
network depends on not only its own SNR but also channel orthogonality
between concurrent transmitters. To illustrate this point, in
Fig.~\ref{fig:angle-tput}, we analytically compute the throughput of
the second client decoded by ZF for the whole range of the
inter-client angle $\theta \in [0,\pi/2]$ when its original SNR at the
AP is 5, 10, 15, 20 and 25 dB, respectively.\footnote{We empirically
measure an SNR-throughput mapping table and map the SNR after
projection to the corresponding throughput.}

\begin{figure}[t!]
\centering
\epsfig{file=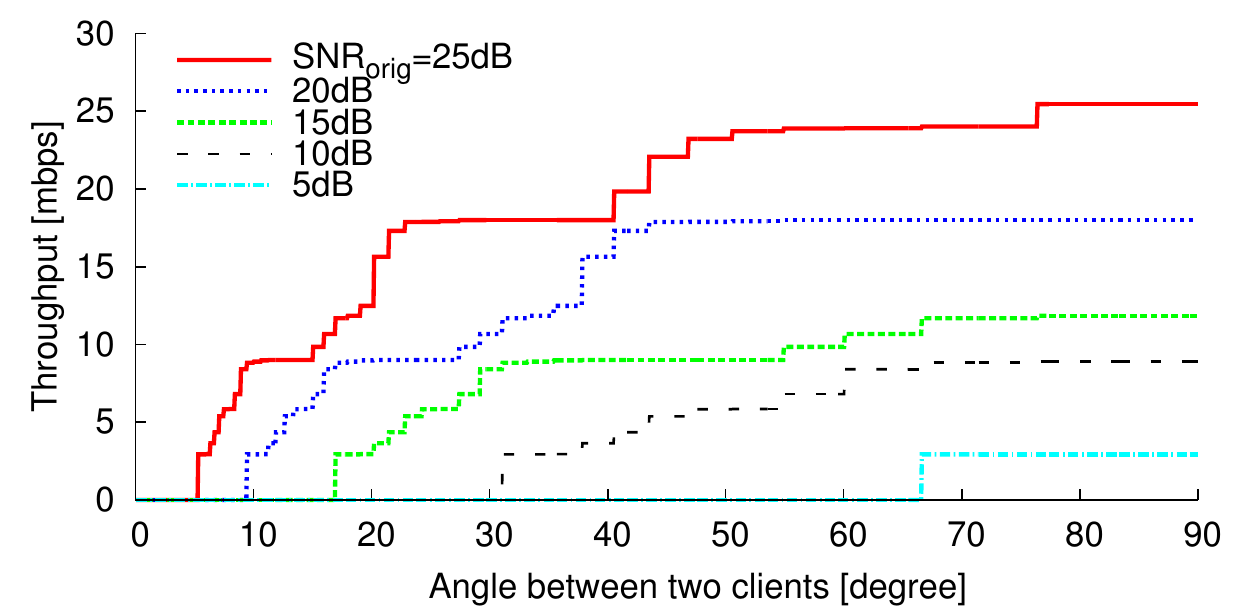, width=3in}
\vspace{-12pt}
\caption{Throughput after projection}
\label{fig:angle-tput}
\end{figure}

The figure shows that a small inter-client angle significantly reduces
the SNR after projection and hence the throughput.  If the client has
a high original SNR, e.g., 25~dB, by selecting the best bit rate
according to the SNR after projection~\cite{TurboRate}, it can still
get a relatively high throughput even after projection.  In contrast,
if the original SNR of the client is low, the client would be very
likely to get zero throughput even if it already uses perfect bit rate
adaptation. This is because a small SNR reduction could make its SNR
after projection become lower than the 802.11 operational SNR region,
i.e., 4~dB.  In particular, if the client's original SNR is 10 dB,
then for any angle smaller than 31 degree, its SNR after projection
drops below 4~dB, leading to zero throughput.  The situation is even
worse if the client's original SNR is only 5~dB.

One important thing worth noting is that the success of decoding the
first client using ZF-SIC relies on the AP being able to decode the
second client correctly. Otherwise, the AP cannot remove the
interfering signal of the second client, and hence also fails to
decode the first client. Also, after removing the interfering signal
from the second client, the first client can select its best bit-rate
according to its original SNR without considering who later joins the
concurrent transmission.  The above constraint hence requires the
second client to give up its transmission opportunity if its SNR after
projection is lower than the 802.11 operational SNR region.  This also
motivates why selecting a suitable subset of clients to transmit
concurrently is important for delivering the gain of MU-MIMO.

\vskip 0.05in \noindent {\bf b) How different is the throughput of a
client when it transmits concurrently with different clients?} We have
shown in Fig.~\ref{fig:angle-tput} that the throughput of a client
could change significantly with the inter-client angle.  We next check
whether the angle between the channels of two clients is actually
randomly distributed between $[0,\pi/2]$.  To validate this point, we
empirically measure how much throughput a client can achieve if it
transmits concurrently with different clients in a real testbed where
6 single-antenna clients, named Tx1-Tx6, contend for communicating
concurrently with a 2-antenna AP.  We repeat the experiment twice with
different random locations of the clients. Experiment 1 locates Tx1
close to the AP, while Experiment 2 locates Tx1 far from the AP. Two
experiments represent the scenarios when Tx1 has a high and low
original SNR, respectively.

\begin{figure}[t!]
\centering
{\footnotesize
\begin{tabular}{cc}
\hspace{-0.1in}
\epsfig{file=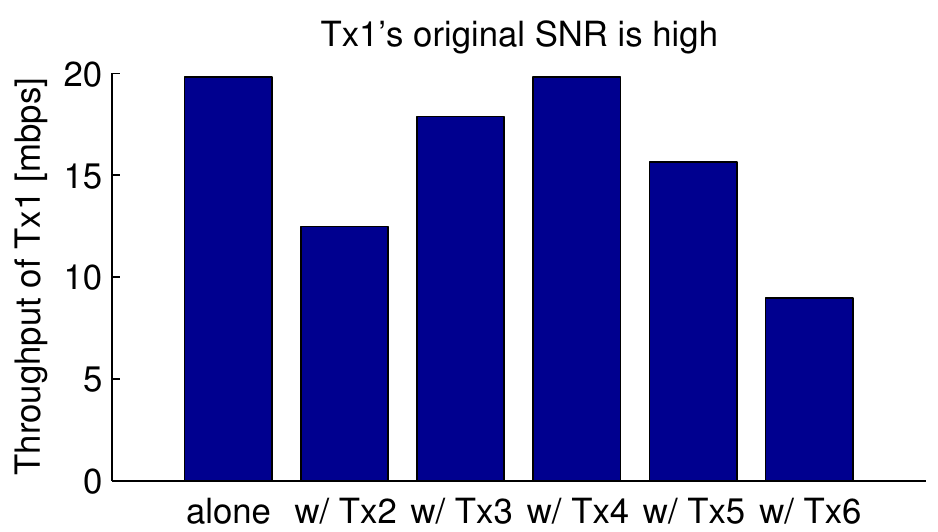, width=1.7in} & 
\hspace{-0.1in}
\epsfig{file=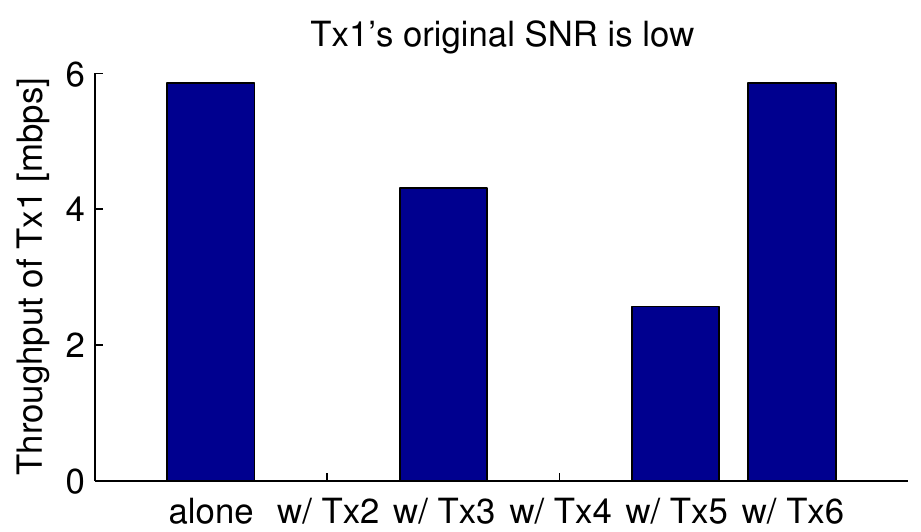, width=1.7in}\\ 
(a) Experiment 1 &
(b) Experiment 2
\end{tabular}
}
\vspace{-12pt}
\caption{Heterogeneous throughput in a MU-MIMO LAN}
\label{fig:diff-tput}
\end{figure}

Fig.~\ref{fig:diff-tput} plots the throughput of one client (denoted
by Tx1) when it transmits alone or when it transmits concurrently with
any of other five clients (denoted by Tx2--Tx6). The figure shows
that, in both experiments, as compared to transmitting alone, Tx1
usually gets a lower throughput when it joins the concurrent
transmission and is decoded by projection. In addition, the client's
throughput, as transmitting concurrently with different users, could
be very different.  For example, in experiment 1, Tx1 obtains a high
throughput when it transmits concurrently with Tx4, while suffering a
low throughput as joining Tx6's transmission. The situation becomes
worse when Tx1 has a low original SNR (as in experiment 2); in many
cases, it gets zero throughput as transmitting concurrently with
another client. These results are consistent with the analysis shown
in Fig.~\ref{fig:angle-tput}.  Thus, to get a high throughput, Tx1
would like to transmit with a client whose channel is more orthogonal
to its channel, as a result experiencing less SNR reduction. The
current random access protocols however do not consider this effect,
and hence cannot efficiently deliver the MU-MIMO gain.

\begin{figure}[t!]
\centering
{\footnotesize
\begin{tabular}{cc}
\hspace{-0.1in}
\epsfig{file=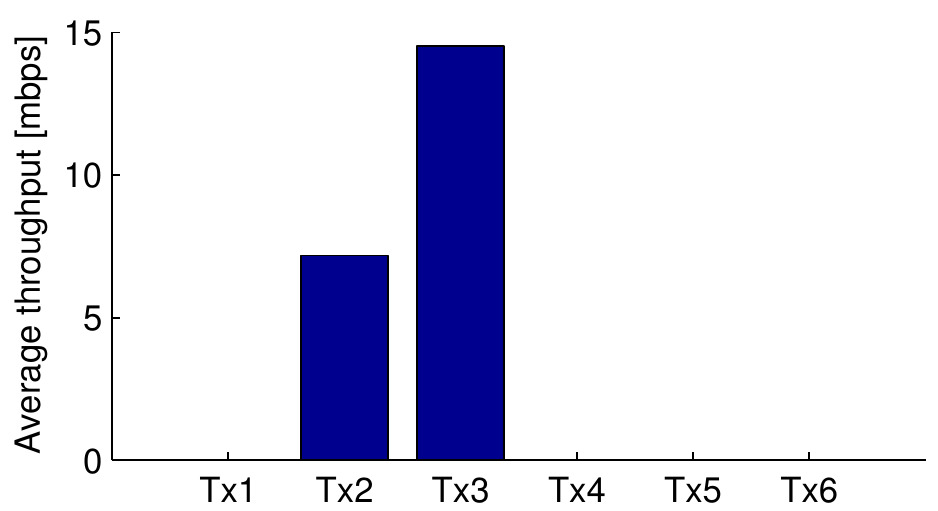, width=1.7in} & 
\hspace{-0.1in}
\epsfig{file=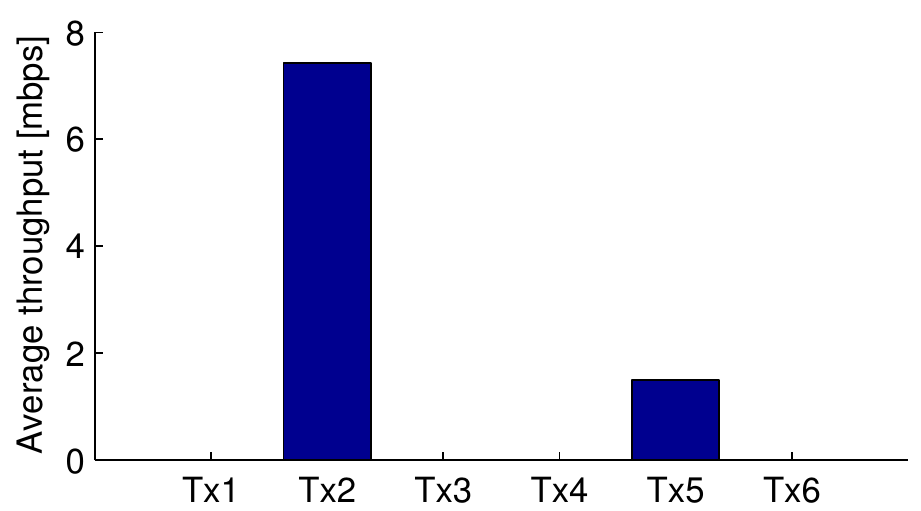, width=1.7in}\\ 
(a) Experiment 1 &
(b) Experiment 2
\end{tabular}
}
\vspace{-12pt}
\caption{Fairness of transmission opportunities
based on the throughput}
\vspace{-10pt}
\label{fig:fair-tput}
\end{figure}

\vskip 0.05in \noindent {\bf c) Is selecting concurrent clients to
maximize the throughput fair enough?}
A na\"{i}ve solution to improving the throughput is to
deterministically assign the client that achieves the highest
throughput to join the transmission of the first contention winner.
This simple solution however might not be fair for all the clients. To
see why this is a problem, we measure the throughput of each client in
the above scheme.  Specifically, in each experiment, we let each
client have an equal probability to transmit the first stream; given
the first contention winner, we then select the client that can
achieve the highest throughput after projection to transmit the second
stream.  Each experiment includes 1,000 rounds of contention, i.e.,
1,000 concurrent transmissions.

Because every client has an equal probability to win the first
contention, we do not consider the throughput of a client transmitting
using the first stream. Instead, in Fig.~\ref{fig:fair-tput}, we only
plot the average throughput of a client transmitting in the second
stream. The results of two experiments show that, by applying such a
na\"{i}ve solution, some clients, e.g., Tx1, Tx4, Tx5 and Tx6 in
experiment 1, do not have any opportunities to join concurrent
transmissions. Even worse, the clients with a low original SNR are
very likely to starve because they cannot compete with those clients
in the high SNR regime.  One might think, alternatively, we can assign
the client that has the largest inter-client angle with the first
contention winner to transmit concurrently. We repeat the same
experiment by applying the above assignment. The results in
Fig.~\ref{fig:fair-ang} show that this solution again fails to provide
fairness because some clients happen to have a small angle with all
the other clients and hence do not get any transmission opportunities.

\begin{figure}[t!]
\centering
{\footnotesize
\begin{tabular}{cc}
\hspace{-0.1in}
\epsfig{file=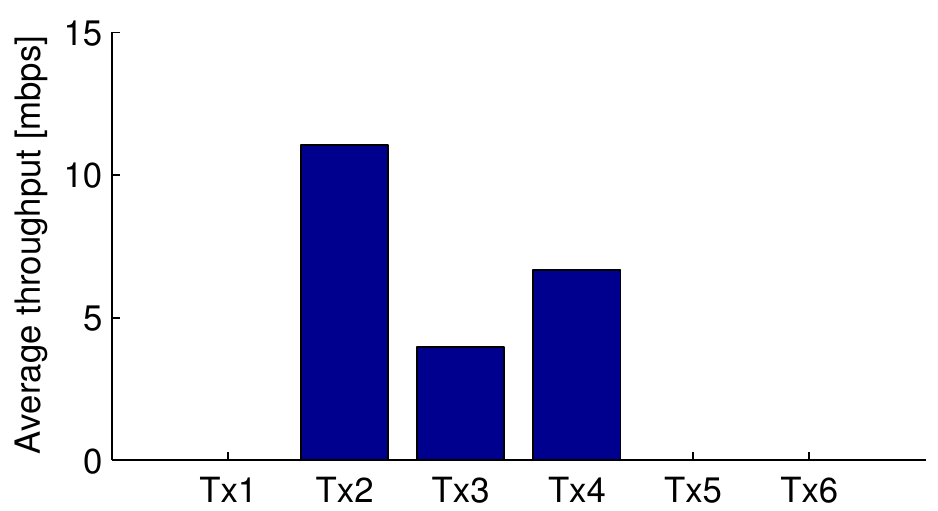, width=1.7in} & 
\hspace{-0.1in}
\epsfig{file=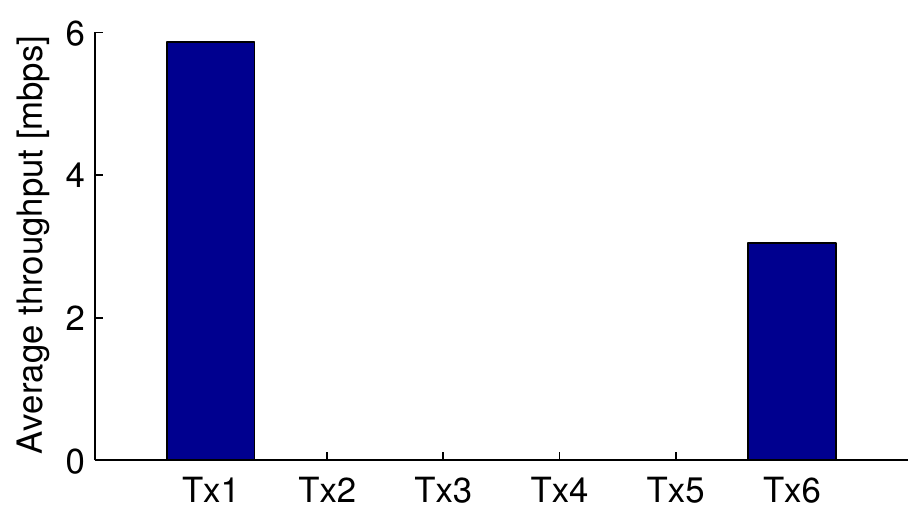, width=1.7in}\\ 
(a) Experiment 1 &
(b) Experiment 2
\end{tabular}
}
\vspace{-12pt}
\caption{Fairness of transmission opportunities
based on the angle}
\label{fig:fair-ang}
\end{figure}
\section{\name\ Matching}\label{sec:match}
Motivated by the above measurements, we aim at designing a matching
protocol, called \name, that pairs concurrent clients to deliver the
maximum throughput gain enabled by concurrent transmissions, while, at
the same time, providing clients fair concurrent transmission
opportunities.  For simplicity, we describe our \name\ protocol
assuming that multiple single-antenna clients communicate concurrently
with a multi-antenna AP in an uplink MU-MIMO LAN.  Our design however
can be generalized to clients with multiple antennas and downlink
MU-MIMO LANs.  We will describe in the next section how to realize
\name\ as a leader-contention-based MAC protocol.

\subsection{Overview}\label{sec:overview}
The goal of \name\ is to build chain relation in concurrent
transmissions. When a client wins contention, the following concurrent
transmissions are determined a priori. Hence, all the concurrent
transmitters only require to precede their streams with one contending
process. In particular, \name\ matches clients whose channels are more
orthogonal to each other as a group of concurrent transmitters with a
precedence relation, which is called {\em MIMO-mates}. To see how it
works, let us consider an example where two clients are allowed to
communicate concurrently with a 2-antenna AP. We match two clients as
{\em MIMO-Mates}, in which one is the lead and the other is the
follower.  When the lead of {\em MIMO-Mates} wins contention and
transmits the first stream, its follower transmits the second stream
concurrently immediately after it detects the transmission from the
lead. 

The protocol can be generalized to an $N$-antenna AP scenario where
$N$ clients can transmit concurrently.  We match $N$ clients as a {\em
MIMO-Mate} (precedence) relation $(u_1,u_2,\cdots,u_N)$ such that
clients $u_1, u_2, \cdots, u_N$ join the concurrent transmissions one
after another in order of precedence.  In particular, after the lead
$u_1$ wins contention, any of the following clients $u_i$ can count
the number of preambles to figure out the time that it should
transmit.  The above protocol benefits throughput gains from two
factors: 1) it matches clients with a higher channel orthogonality to
transmit concurrently and minimizes throughput reduction caused by
projection; 2) it requires only one contending process for concurrent
transmissions, as a result reducing the overhead significantly.
However, the benefit of \name\ might not be able to be fully delivered
when any of the matched {\em MIMO-Mates} does not have traffic to
send. In this case, the unused transmission opportunities should be
exploited by other clients to avoid waste.  We thus further integrate
\name\ with an angle-based contention mechanism, which will be
discussed in~\xref{sec:mac}.

\subsection{Problem Formulation}\label{sec:problem}
Our objective is to match clients as {\em MIMO-mates} in order to
maximize the throughput subject to the fairness constraint.  We first
define our problem in a 2-antenna AP scenario, and next extend it to a
3-antenna AP scenario and even a more general $N$-antenna AP scenario.

Let us first consider the 2-antenna AP scenario.  Say $u$ is the client
that wins the first contention, and $v$ is the follower of $u$, who
joins $v$'s transmission. We define $(u,v)$ as the {\em MIMO-mate}
relation of clients $u$ and $v$. Let $r^{(u,v)}_v$ denote the
throughput of $v$ as it transmits concurrently with $u$ and is decoded
by using ZF to project orthogonal to client $u$.  We note that client
$v$ might get a different throughput if it is assigned to follow a
different predecessor, i.e., $r^{(u,v)}_v$ could be different from
$r^{(u',v)}_v$ if $u \neq u'$.  The {\em MIMO-mate} matching problem in
a 2-antenna AP scenario can be defined as follows:

\begin{probdefi}\label{Probdefi: 2mimomate} ({\bf \em 2-MIMOMate})
Given a set of clients $V$ and the throughput $r^{(u,v)}_v$ for all
$u,v \in V$, the matching problem is to find a set $M \subseteq
V\times V$ such that
\begin{enumerate}
\item \label{Cndtn: 1-2mimomate} $r^{(u,v)}_v>0, \forall (u,v) \in M$,
\item \label{Cndtn: 2-2mimomate} $u \neq v, \forall (u, v) \in M$,
\item \label{Cndtn: 3-2mimomate} $u_1 \neq u_2$ and $v_1 \neq v_2$
	for any two distinct elements $(u_1, v_1), (u_2, v_2) \in M$,
\item \label{Cndtn: 4-2mimomate} $\vert M \vert$ is maximized,
\item \label{Cndtn: 5-2mimomate} $\sum_{(u,v) \in M} r^{(u,v)}_v$ is
	maximized among those $M$ satisfying Constraint~\ref{Cndtn:
	4-2mimomate}.
\end{enumerate} 
\end{probdefi}

To ensure the success of ZF-SIC decoding, Constraint~\ref{Cndtn:
1-2mimomate} allows a client to join the concurrent transmission only
if it can be successively decoded, i.e., getting a positive
throughput.  Constraints~\ref{Cndtn: 2-2mimomate}--\ref{Cndtn:
3-2mimomate} force each client to follow at most one of other clients,
and hence guarantee fairness.  The rationale of Constraint~\ref{Cndtn:
4-2mimomate} is that, since each client in traditional 802.11 has an
equal probability to win the first contention and transmit the first
stream, then, by finding the maximum set $M$, we allow as many clients
as possible to join the concurrent transmission.  This ensures clients
to also have a fair probability to transmit the second stream. Under
such a fairness constraint, our goal is to find a feasible solution
that maximizes the total throughput of the followers, i.e., the second
streams. Note that we are only interested in the throughput of the
followers because, by using ZF-SIC, a client who transmits the first
stream can get about the same throughput no matter who its follower
is~\cite{TurboRate}.

\begin{figure}[t!]
\centering
{\footnotesize
\begin{tabular}{cc}
\epsfig{file=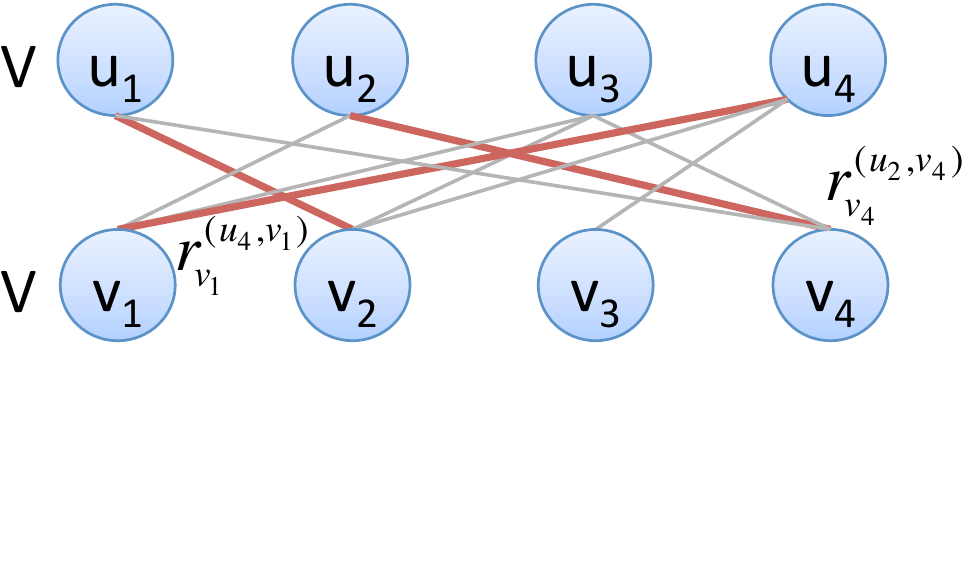, scale=.4} &
\epsfig{file=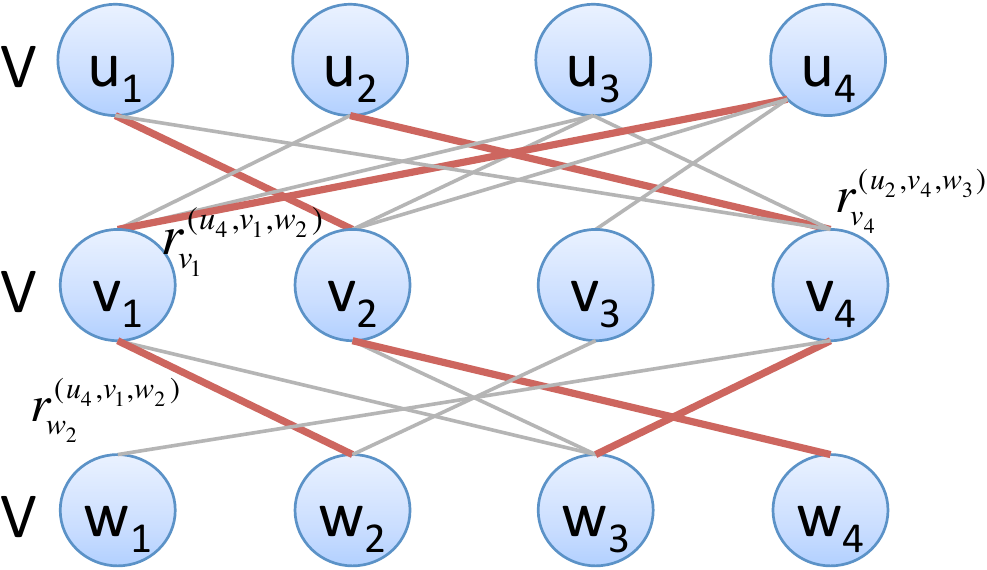, scale=.4}\\
(a) 2-{\em MIMOMate} matching &
(b) 3-{\em MIMOMate} matching 
\end{tabular}
}
\vspace{-12pt}
\caption{{\em MIMOMate} matching for a network with
4 clients}
\label{fig:graph}
\end{figure}

The above {\em 2-MIMOmate} problem can actually be illustrated as a
bipartite graph, as shown in Figure~\ref{fig:graph}(a). Each edge
$(u,v)$ is associated with a weight, which is set to the throughput of
$v$ when it follows $u$, i.e., $r^{(u,v)}_v$.  We observe that the
{\em 2-MIMOmate} problem is exactly equivalent to the {\em Bipartite
Maximum Weighted Maximum Cardinality Matching} problem, which finds a
maximum cardinality matching with maximum weight in a bipartite graph
and can be solved in polynomial time by the algorithm proposed
in~\cite{LEDA} (see Chapter 7.8). Note that, since \name\ always
assigns any leader a specific follower, any feasible solution of
Problem~\ref{Probdefi: 2mimomate} is therefore deterministic.  That
is, the relationship between the first contention winner and its
follower is fixed until we solve Problem~\ref{Probdefi: 2mimomate}
again when the channels change.  This is very different from the
probabilistic nature of \textit{uniformly random} contention used in
most of existing protocols~\cite{SAM,multiround, multiround2}, where
every client has the same probability to win the contention of sending
concurrently with a given first winner.  A natural question then
arises: Would choosing the follower uniformly randomly, e.g., via {\em
uniformly random} contention, results in a solution better than the
output of Problem~\ref{Probdefi: 2mimomate}?  We give the following
positive result on Problem~\ref{Probdefi: 2mimomate}:
\begin{thrm}
\label{thrm: mimomate_vs_contention}
If the throughput $r^{(u,v)}_v$ is greater than $0$ for all $u,v \in
V, u \neq v$, then the average throughput of Problem~\ref{Probdefi:
2mimomate}'s output is higher than or equal to the average throughput
achieved by using \textit{uniformly random} contention to choose the
follower, even if the contention overhead is ignored.
\end{thrm}
In fact, we can prove the following stronger theorem, which replaces
the \textit{uniformly random} contention by any probabilistic
assignment that obeys the fairness constraint.  More specifically, in
any fair probabilistic assignment, every client has the same
probability to transmit the second stream; however, any first winner
might choose its follower with a non-uniform probability, and
different leaders might have different probability distributions.
\begin{thrm}
\label{thrm: mimomate_vs_random}
If the throughput $r^{(u,v)}_v$ is greater than $0$ for all $u,v \in
V, u \neq v$, then the average throughput of Problem~\ref{Probdefi:
2mimomate}'s output is higher than or equal to the average
throughput achieved by any fair probabilistic assignment.
\end{thrm}
We prove the two theorems in the Appendix.B. Note that the two theorems
only holds when the throughput $r^{(u,v)}_v$ is greater than $0$ for
all $u,v \in V, u \neq v$. If this condition does not hold, i.e., some
$r^{(u,v)}_v = 0$, then using contention to choose the followers,
e.g., the method used in~\cite{SAM, multiround, multiround2}, would
fail ZF-SIC decoding.  On the other hand, to guarantee the success of
ZF-SIC decoding,  such a pair of clients would not be chosen in
Problem~\ref{Probdefi: 2mimomate}.  In addition, if we further
consider the overhead of using contention to choose followers, the
throughputs of the methods proposed in~\cite{SAM, multiround,
multiround2} would further decrease.

We next consider the 3-antenna AP scenario. Say clients $u,v$ and $w$
communicate with a 3-antenna AP concurrently and join the concurrent
transmissions one after another.  We define $(u,v,w)$ as the {\em
MIMO-mate} relation of clients $u,v$ and $w$.  The AP can use ZF-SIC
to decode client $w$ by projecting along the direction orthogonal to
both clients $u$ and $v$.  It re-encodes client $w$'s stream and
subtracts it from the received signals. The AP then decodes client $v$
by projecting the resulting signal along the direction orthogonal to
client $u$, and decodes client $u$ after removing the signals of
clients $v$ and $w$.  Let $r^{(u,v,w)}_v$ and $r^{(u,v,w)}_w$ denote
the throughput of clients $v$ and $w$, respectively. The {\em
MIMOMate} matching problem in a 3-antenna AP scenario can be defined
as follows:

\begin{probdefi}\label{Probdefi: 3mimomate} ({\bf \em 3-MIMOMate}) Given a
set of clients $V$ and the throughput $r^{(u,v,w)}_u$ and
$r^{(u,v,w)}_w$ for all $u,v,w \in V$, the matching problem is to find
a set $M \subseteq V \times V \times V$ such that
\begin{enumerate}
\item \label{Cndtn: 1-3mimomate} $r^{(u,v,w)}_v>0$ and
	$r^{(u,v,w)}_w>0 , \forall (u,v,w) \in M$,
\item \label{Cndtn: 2-3mimomate} $u \neq v \neq w, \forall (u, v, w) \in M$,
\item \label{Cndtn: 3-3mimomate} $u_1 \neq u_2, v_1 \neq v_2$, and
	$w_1 \neq w_2$ for any two distinct elements $(u_1, v_1, w_1),
	(u_2, v_2, w_2) \in M$,
\item \label{Cndtn: 4-3mimomate} $\vert M \vert$ is maximized,
\item \label{Cndtn: 5-3mimomate} $\sum_{(u,v,w) \in M} 
	(r^{(u,v,w)}_{v}+r^{(u,v,w)}_w)$ is maximized among those $M$
	satisfying Constraint~\ref{Cndtn: 4-3mimomate}.
\end{enumerate} 
\end{probdefi}

\begin{algorithm}[t]
\label{algo}
\begin{small}
\DontPrintSemicolon
\SetKwInOut{Input}{input}
\Input{a set of clients $V$ and the channel state information from each
client to AP's $N$ antennas}

Duplicate $V$ to $V_1, V_2,\cdots, V_N$\; 
Remove legacy 802.11 nodes from $V_2,\cdots,V_N$\; 
Initialize $M\leftarrow \{\}$

\For {$k:=1$ to $N-1$} { 
	For each edge $(u_i,v_j){\in}V_k \times V_{k+1}$, if $u_i$ has a
	predecessor or $u_i{\in}V_1$,  set the weight of edge $(u_i,v_j)$
	to the throughput of $v_j$ as it transmits concurrently with $u_i$
	and all its predecessors;
	otherwise, set the weight of $(u_i,v_j)$ to 0 \;
	$M'\leftarrow$ the solution of the 2-{\em MIMOMate} matching
	problem for $V_k \times V_{k+1}$ solved by~\cite{LEDA}\;
	\If {$k = 1$} {
		Add each $(u_i,v_j){\in} M'$ to $M$\; 
	}
	\Else {
		For each $(u_i,v_j){\in}M'$,  find the {\em MIMO-Mate}
		relation (element) $m \in M$ that includes $u_i \in V_k$ and
		add client $v_j \in V_{k+1}$ to the element $m$ \; 
	}
	
}
\Return $M$\;
\caption{$N$-{\em MIMOMate} Matching Algorithm}
\end{small}
\end{algorithm}

Similarly, in the 3-{\em MIMOMate} problem, we are only interested in
maximizing the throughput of the followers, i.e., clients $v$ and $w$.
We observe that the 3-{\em MIMOMate} problem, as illustrated in
Figure~\ref{fig:graph}(b), is actually a variation of the {\em Maximum
3-Dimensional Matching} problem~\cite{GareyJohnson}, which is defined
as follows: Let $X$, $Y$, and $Z$ be disjoint sets, and let $T$ be a
subset of $X {\times} Y {\times} Z$ that includes all feasible
matching combinations.  The problem finds the maximum matching $M
\subseteq T$ such that $u_1 {\neq} u_2, v_1 {\neq} v_2$ and $w_1
{\neq} w_2$ for any two distinct elements $(u_1, v_1, w_1)$ and $(u_2,
v_2, w_2)$ in $M$.  Hence, the differences between our 3-{\em
MIMOMate} problem and the 3-dimensional matching problem are 1) our
problem further considers the total weight of a matching (i.e.,
Constraint~\ref{Cndtn: 5-3mimomate}), and 2) Constraint~\ref{Cndtn:
2-3mimomate} in our problem restricts each client to be included in an
element at most once.  For example, $(u, u, v)$ is not a feasible
combination in our problem because client $u$ cannot transmit two
streams from its single antenna at the same time.

We note that a general $N$-{\em MIMOMate} matching problem can be
formulated in a similar manner, and it is a variation of the
$N$-dimensional matching problem~\cite{GareyJohnson}. On the other
hand, although the 2-{\em MIMOMate} problem is polynomial time
solvable, the $N$-{\em MIMOMate} problem  for any $N \ge 3$ is however
NP-hard. We will prove the NP-hardness of the 3-{\em MIMOMate} problem
in the Appendix.A by deriving a reduction from the 3-dimensional
matching problem (which is also NP-hard) to our problem.  The
NP-hardness of the $N$-{\em MIMOMate} problem can be proved in a
similar way.

\subsection{Heuristic Matching Algorithm}\label{sec:algo}
There is an approximation algorithm~\cite{3dm} proposed to solve the
$N$-dimensional matching problem. We can use the algorithm to solve
our $N$-{\em MIMOMate} matching problem and achieve an approximation
ratio, 3/2+$\epsilon$, for any $\epsilon > 0$, in terms of the size of
matching. It however does not ensure to find the one achieving the
maximal throughput (i.e., Constraint~\ref{Cndtn: 5-3mimomate}) among
all maximum matchings.  In addition, our problem requires an
additional cost to compute the weights (throughputs) of all possible
{\em MIMO-Mates}, which is an $O(|V|^N)$ computational cost.  We hence
propose an algorithm, as shown in Algorithm~\ref{algo}, to solve our
{\em MIMOMate} matching problem with a reduced cost of weight
computation.  The basic idea of Algorithm~\ref{algo} is to decompose
the $N$-{\em MIMOMate} matching problem into $(N-1)$ 2-{\em MIMOMate}
matching problems, each of which can be solved by the bipartite
maximum weighted maximum cardinality matching algorithm~\cite{LEDA} in
polynomial time. The advantage of such decomposition is that it
reduces the cost of weight computation from $O(|V|^N)$ to $O(N|V|^2)$.  

\begin{figure}[t!]
\centering
{\footnotesize
\begin{tabular}{cc}
\epsfig{file=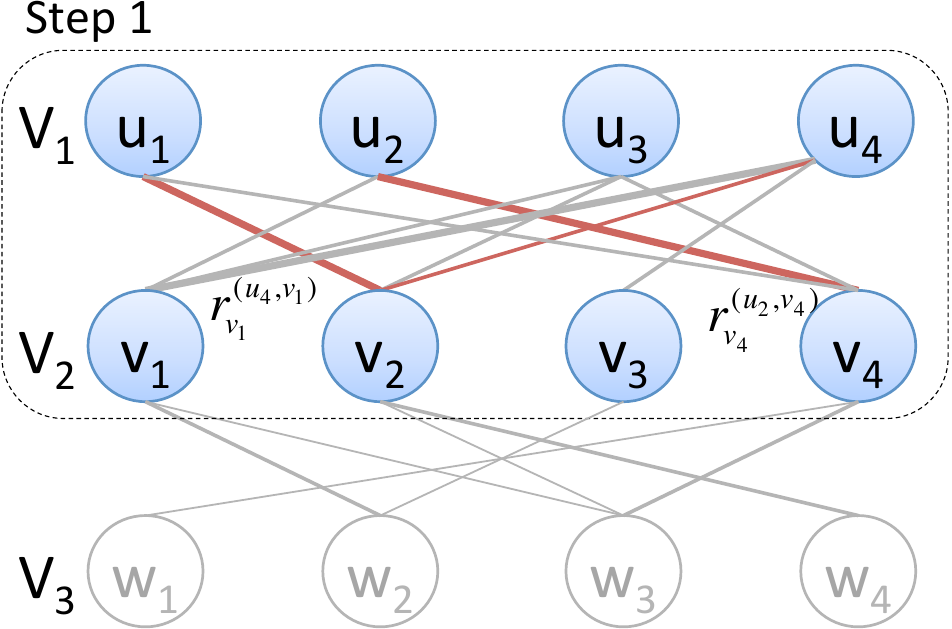, scale=.42} &
\epsfig{file=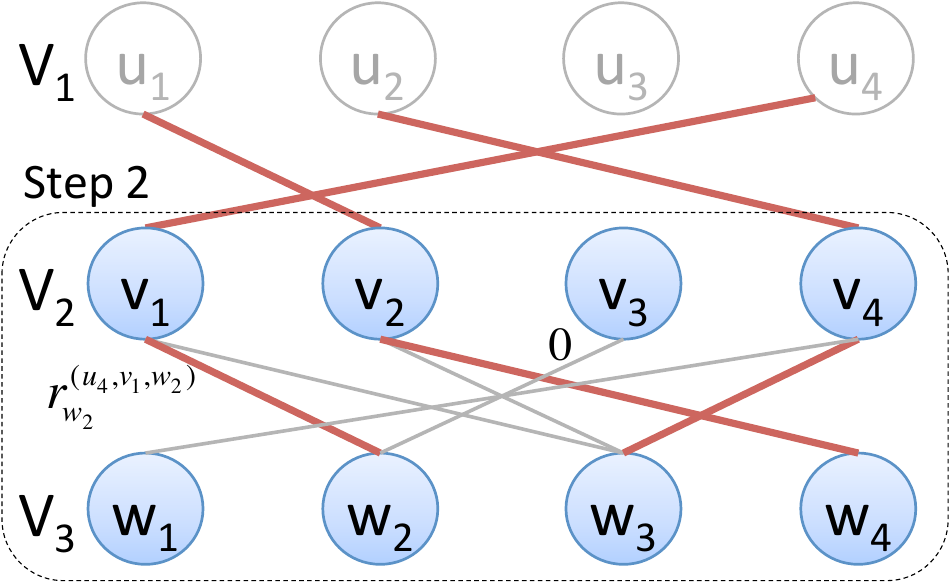, scale=.42}\\
(a) step 1 &
(b) step 2
\end{tabular}
}
\vspace{-12pt}
\caption{Example of solving 3-{\em MIMOMate} Matching}
\label{fig:3matching}
\end{figure}

For simplicity, we use the 3-antenna AP scenario to describe our
algorithm, and next explain how to generalize it to an $N$-antenna AP
scenario. We first duplicate the client set $V$ to $V_1, V_2$ and
$V_3$ (line 1 in Algorithm~\ref{algo}).  Our algorithm solves the
3-{\em MIMOMate} matching problem in two steps, as illustrated in
Fig.~\ref{fig:3matching}. In the first step, as in
Fig.~\ref{fig:3matching}(a), we set the weight of edge
$(u_i,v_j){\in}V_1{\times}V_2$ by computing the throughput of $v_j
{\in} V_2$ as it transmits concurrently with $u_i {\in}V_1$, and try
to optimally match any second client in $V_2$ to a {\em MIMO-Mate}
lead in $V_1$, which is actually the 2-{\em MIMOMate} matching problem
for $V_1 \times V_2$.  After the first step, we determine the {\em
MIMO-Mate} relation for the first and second transmitters.  Then, in
the second step, our goal is to add any third transmitter in $V_3$
into each {\em MIMO-Mate} relation. To do so, we first assign edge
$(v_j, w_k){\in}V_2{\times}V_3$ a weight that equals the throughput of
$w_k$ when it transmits concurrently with $v_j$ and $v_j$'s
predecessor, which is solved in the first step. For example, in
Fig.~\ref{fig:3matching}(b), because $(u_4, v_1)$ is matched as {\em
MIMO-Mates} in the first step, the weight of edge $(v_1, w_2)$ equals
the throughput of $w_2$ when it joins concurrent transmissions of
$u_4$ and $v_1$. We note that, for any $v_j{\in}V_2$, if it is not
assigned a predecessor in the first step, we set the weight of all its
outgoing edges to 0 because it is not allowed to match with any
clients in $V_3$ in the second step.  After weight assignment, we can
solve another 2-{\em MIMOMate} matching problem to match clients in
$V_3$ to clients in $V_2$.

Our algorithm can be generalized to the $N$-{\em MIMOMate} matching
problem.  Specifically, it iteratively solves a 2-{\em MIMOMate}
matching problem to match clients in $V_{k+1}$ to clients in $V_k$,
where $k=1,2,\cdots,N-1$.  Hence, each iteration includes a client in
a {\em MIMO-Mate} relation. In particular, the $k^{th}$ iteration adds
the $(k+1)^{th}$ concurrent client to {\em MIMO-Mates}. In addition,
in the $k^{th}$ iteration, because we already know the first $k$
clients in {\em MIMO-Mates}, we thus only need to compute the weight
of edges $(u,v) \in V_k \times V_{k+1}$ according to the given {\em
MIMO-Mate} relation. This is why our algorithm can reduce the cost of
weight computation to $O(N|V|^2)$.

So far we describe our algorithm by assuming that all the clients are
\name\ nodes.  Our algorithm can be slightly adjusted to allow the
coexistence of \name\ nodes and legacy 802.11 nodes.  Recall that we
duplicate the client set $V$ to $V_1, V_2, \cdots, V_N$, where $N$ is
the number of antennas equipped on the AP.  Each node in the
duplicated set $V_i$ is a candidate of sending the $i^{th}$ stream.
Note that legacy nodes follow the traditional 802.11 operation and can
only contend for sending the first stream. In other words, legacy
nodes do not leverage the concurrent transmission opportunities and
will not follow any ongoing transmissions.  We can hence simply remove
those legacy nodes from $V_2,\cdots,V_N$, and only keep them in $V_1$.
By doing this, legacy nodes can still use conventional 802.11
contention to occupy the first dimension, and can further be followed
by some other \name\ nodes. Consider Fig.~\ref{fig:3matching} as an
example. Assume that node $u_4$ is a legacy node. We only put it in
$V_1$, but not in $V_2$ and $V_3$. It hence can be followed by some
other MIMOMate nodes, but cannot join concurrent transmissions.
\section{\name's Medium Access Protocol}\label{sec:mac}
We consider a MU-MIMO MAC protocol similar to SAM~\cite{SAM}, where
clients join the concurrent transmissions one after another.  Like
SAM~\cite{SAM}, clients join concurrent transmissions one
after another. Each client counts the number of concurrent streams by
cross-correlating with the known preamble in the presence of
ongoing transmissions. Clients can join the concurrent transmissions
until they detect that the number of existing streams equals the
number of antennas at the AP.\footnote{The preamble-counting based
	protocol, like SAM~\cite{SAM}, could suffer from collisions when
	hidden nodes interrupt the preamble-counting process. We apply a
	multi-round light-weight handshaking mechanism proposed
	in~\cite{lightrtscts} to address the hidden terminal problem with
minimum overhead.}
Each client determines its best bit-rate based on TurboRate, the
MU-MIMO rate adaptation scheme proposed in~\cite{TurboRate}.
TurboRate allows each client to announce training symbols before data
transmission.  All the clients who contend for the later transmissions
can hence learn the channels of the ongoing streams from those
training symbols, and adapt the bit-rates based on their channels.
Moreover, TurboRate asks clients to give up contention opportunities
if their SNR after ZF-SIC decoding is too low to be decodable.  To
increase	the gain of MU-MIMO, the protocol forces concurrent
clients to end their transmissions at about the same time. To do so,
concurrent clients overhear the information about the frame duration
of the first stream, which is embedded in the MAC header, and fragment
or aggregate their packets accordingly~\cite{TurboRate}\cite{n+}.
\name\ differs from the existing MAC protocols in that it only allows
clients to use 802.11's CSMA/CA to contend for the first stream, but
lets the remaining clients join the concurrent transmission of its
predecessor in the {\em MIMO-Mate} relation scheduled by
Algorithm~\ref{algo}.  In particular, say a client is scheduled to
transmit the $k^{th}$ stream in the {\em MIMO-Mates}; it can start
transmitting once it detects $k-1$ preambles from all its predecessors
after its leader wins the contention.  Hence, all clients in the {\em
MIMO-Mates} only require one contending process.

To realize such a user matching protocol, \name's MAC needs tow major
modifications: 1) the AP needs to learn the uplink channel information
of all its clients, and 2) the AP needs to announce the matching
result to its clients.  To learn the channel information, one possible
solution is to let all the clients learn their uplink channels and
report this information to the AP. To do so, the clients leverage {\em
channel reciprocity}~\cite{reciprocity}, which refers to the property
that the channels in the forward and reverse directions are the same.
Using reciprocity, every client can exploit the beacons to learn the
downlink channel and use it to estimate the uplink channel.  It is
however an expensive overhead to ask all the clients to report their
channels for every packet transmission. On the other hand, legacy
nodes, which follow the traditional 802.11 operation, do not feedback
this information.  We hence perform the following optimizations to
reduce the overhead of channel feedback: The AP learns the uplink
channel of a client from its uplink frames, including the association
frames when that client joins the network, the data frame of its
uplink packets, and the ACK of its downlink packets. The AP hence only
needs to re-schedule {\em MIMO-Mates} when it detects that the
channels of certain clients change due to channel variation or user
mobility.  Once the AP reschedules {\em MIMO-Mates}, it announces the
updated matching result to the \name\ nodes. A simple solution is to
annotate the periodical beacon messages with the announcement.
However, legacy nodes might not be able to identify the modified
beacon format. To enable the coexistence of \name\ nodes and legacy
nodes, we let the AP send the matching result in another control frame
using the subtype not used in conventional
802.11~\cite{80211_OReilly}.

\begin{figure}[t!]
\centering
\epsfig{file=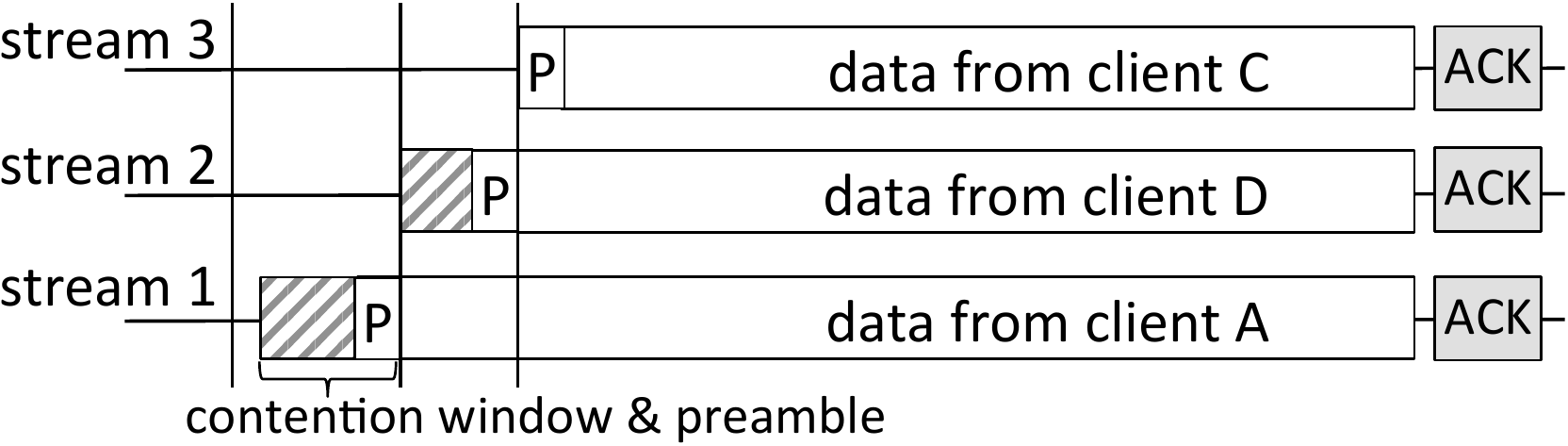, scale=.4}
\vspace{-12pt}
\caption{Angle-based contention}
\label{fig:protocol}
\end{figure}

We notice that the above protocol only operates properly if the
scheduled {\em MIMO-Mates} always have traffic to transmit. However,
in practice, a client might have a bursty traffic pattern. Hence, to
enable users to fully utilize concurrent transmission opportunities,
we can further integrate \name\ with any contention-based MAC
protocol. Specifically, clients can contend for the unused degrees of
freedom if the scheduled {\em MIMO-Mates} do not have traffic to
transmit.  However, to better exploit the gain of concurrent
transmissions, we propose to integrate \name\ with an {\em
angle-based} contention scheme. In particular, when any of the
scheduled {\em MIMO-Mates} does not have traffic to transmit, we allow
other clients to contend for the concurrent transmission opportunity,
e.g., the second stream in the example shown in
Fig.~\ref{fig:protocol}, until the number of concurrent streams equals
the number of antennas supported by the AP, $N$.  However, we modify
the contention mechanism to assign different users a different
probability of winning a concurrent transmission opportunity,
according to their channel orthogonality with the ongoing streams.
Specifically, we tend to let a client with a larger angle between its
channel and the channels of the ongoing streams have a higher
probability to win the concurrent transmission opportunity such that
SNR reduction due to projection can be minimized. 

\begin{algorithm}[t]
\label{algo:contention}
\begin{small}
\DontPrintSemicolon
\SetKwInOut{Input}{input}
\Input{the initial contention window of the $k^{th}$ stream $CW^k
\gets CW_{min}$; the initial update of the $k^{th}$ stream
$\delta^k_{cur} \gets 0$; $N$ antennas at the AP}

\For {$1 \le k \le N$} {
	\textsf{\footnotesize{//contention for the $k^{th}$ stream in each packet
	transmission}}\\
	Learn the angle $\theta \in [0,\pi/2]$ between the client's
	channel and the channels of the $(k-1)$ ongoing streams\;
	\If {\mbox{SNR after projection} $\leq$ \mbox{802.11 
	SNR regime}}
	{
		{\bf return;} \textsf{\footnotesize{//give up this contention}}
	}
	Update $CW^k$ using traditional 802.11 backoff\; 
	\If {$k>1$} {
		$\delta^k_{last} \gets \delta^k_{cur}$\;
		$CW_{orig}^k \gets CW^k  - \delta^k_{last}$\;
		$CW^k \gets CW_{orig}^k - \frac{\theta - \pi/4}{\pi/4}*CW_{orig}^k$\;
		$CW^k \gets \max(CW_{min}, \min(CW_{max}, CW^k))$\;
		$\delta^k_{cur} \gets CW^k - CW^k_{orig}$\;
	}
}
\caption{Angle-based Contention Scheme}
\end{small}
\end{algorithm}

To achieve this goal, we apply Algorithm~\ref{algo:contention} to
adjust the contention window for each concurrent stream according to
the channels of the concurrent transmitters. Specifically, when a node
contends for sending the $k^{th}$ stream in the presence of $(k-1)$
ongoing transmissions, it will adjust its contention windows based on
the angle between its own channel and the channels of the $(k-1)$
ongoing transmitting clients. We assume that clients can learn the
angle between its own channel and the channels of the ongoing
transmitters using the distributed method proposed
in~\cite{TurboRate}. The high-level idea of the angle-based contention
is that, if this angle is large, we let the client decrease its
contention window and hence earn a higher probability to win the
concurrent transmission opportunity.  Otherwise, the client gives
other clients a higher priority to transmit concurrently by increasing
its contention window.

To realize the above design, we let each client maintain a distinct
contention window $CW^k$ for the contention of the $k^{th}$ stream.
The contention windows are adjusted according to the channels of the
ongoing clients. The amount of increment (or decrement) is
proportional to the inter-client angle, i.e.  $\frac{\theta - \pi/4
}{\pi/4}$ in line 10.  To ensure fairness, we ask a client assigned a
higher priority in the current packet to pay back its opportunity in
the next packet.  To this end, if a client decreases (increases) the
contention window by $\delta$ for the current packet, it pays (earns)
the priority back by increasing (decreasing) $\delta$ to its
contention window for the next packet, i.e., $\delta_{last}$ in line
8. The above contention scheme can be applied for the contention of
each concurrent stream until the number of streams reaches the number
of antennas supported by the AP, i.e., $1 \le k \le N$.

\vskip 0.05in
\noindent {\bf Overhead and complexity:}
Recall that implementing \name\ as a MAC protocol relies on two
modifications: 1) the AP needs to learn the uplink channels, and 2)
the AP needs to announce the matching results. Note that we let the AP
measure the channel information from historical uplink frames without
any additional message overhead.  The only additional overhead
required by our design is the matching announcement.  We will show in
\xref{sec:result} that such a small overhead does not offset the gain
of our matching algorithm.  On the other hand, since our matching
algorithm is performed in the access point, and the complexity of
clients should not change much.  Therefore, the only supports we need
from \name\ clients are that 1) they need to receive the matching
announcement, and 2) they need to adapt the contention window size
based on a simple operation defined in our angle-based contention
scheme, as in Algorithm 2. We believe additional power consumption in
the clients due to our design should be negligible.

\section{Experimental Results}\label{sec:result}
We build a prototype of {\name} using the USRP-N200 radio platform,
which is equipped with an RFX2400 daughterboard.  A multi-antenna AP
is built by combining multiple USRP-N200 boards using an external
clock. We implement an OFDM PHY layer with standard 802.11 modulations
(BPSK, 4-64QAM) and code rates. Since USRP-N200 operates on a 10MHz
channel, the possible bit rates range from 3 to 27 Mb/s.  We evaluate
the performance of \name\ in both 2-antenna and 3-antenna AP
scenarios.  Limited by the number of USRPs we have, we set 6 clients
to contend for transmitting two packets concurrently to the 2-antenna
AP, while setting 5 clients to contend for transmitting three packets
concurrently to the 3-antenna AP.  To allow multiple clients to
transmit concurrently, we leverage the synchronization method used
in~\cite{TurboRate}\cite{n+}. Specifically, for each experiment, the
AP broadcasts a trigger signal.  Each client records the timestamp of
detecting the trigger, $t_{trigger}$, waits a pre-defined period of
time, $t_{\Delta}$, and sets the timestamp of the beginning of its
transmission to $t_{start} = t_{trigger} + t_{\Delta}$. In our
testbed, $t_{\Delta}$ is set to 0.1s, which is long enough to tackle
the delays introduced by software.

We compare the following schemes: 1) \name, which is our proposed
protocol, 2) {\em max-throughput first}, which always allows the
client that achieves the maximal throughput after projection to join
the concurrent transmissions, 3) {\em max-angle first}, which always
allows the client that has the maximum angle with the ongoing
transmissions to transmit concurrently, 4) SAM~\cite{SAM}, i.e., {\em
contention-based} protocol without RTS/CTS, which assigns all users an
equal probability to sequentially contend for each concurrent
transmission opportunity, and 5) MRC~\cite{multiround}, i.e.,
multi-round contention, which also assigns each client an equal
probability of winning contentions, but precedes concurrent
transmissions with multiple rounds of RTS and a single CTS. For all
the comparison schemes, we apply TurboRate~\cite{TurboRate}, a MU-MIMO
rate adaptation scheme, to allow concurrent clients to select their
best bit rates.

Due to the timing constraints limited by software radio, we do not
implement contention, random backoff and ACK in USRPs.  Instead, for
each experiment, we offline create a packet trace of 1,000 1500-byte
packets for each client.  The traces of different clients are
generated based on the above four comparison schemes, and ensure that
there are at most 2 and 3 clients assigned to transmit concurrently in
a particular time-slot in 2- and 3- antenna AP scenarios,
respectively. In particular, in the beginning of each experiment, we
let each client transmit training symbols, one after another, for the
AP to estimate its uplink channel. The AP then performs offline
contention to generate 1,000 rounds of concurrent transmissions.  For
all the comparison schemes, in each round of concurrent transmissions,
the AP assigns each client a randomly-selected backoff value between 1
and its contention window, and picks the client with the smallest
backoff value to send the first stream. The contention window of each
client is updated according to the 802.11 standard if collisions
occur. The AP then assigns the remaining concurrent transmission
opportunities to other clients based on design of various comparison
protocols. For example, in SAM, the AP uses the same contention scheme
to assign the remaining transmission opportunities; in \name\, the AP
assigns {\em MIMO-Mates} of the first contention winner to transmit
concurrently; in {\em max-throughput first} and {\em max-angle first},
the AP assigns the remaining transmission opportunities based on the
throughput and the angle between channels, respectively. Based on the
contention results, the AP generates the packet trace for each client,
and immediately sends the trace to each client through Ethernet
connection. Each USRP client can hence read its offline-generated
packet trace and determine the time that it should transmit packets
accordingly.  Clients are asked to send null symbols, i.e., 0, if they
are not selected to transmit in a round of packet transmissions. Since
offline contention performed in the AP does not take too much time, we
expect that the channels do not change significantly, i.e., the
channels during data transmissions would be similar to those learned
in the training phase. In addition, since USRPs cannot implement
real-time ACK, we disable retransmissions in the experiments.  That
is, the AP simply drops a packet if the packet cannot be received or
decoded correctly.

We first evaluate the performance of \name\ when clients have a
continuous traffic pattern, and next evaluate the performance of
integrating the angle-based contention mechanism
(Algorithm~\ref{algo:contention}) with \name\ when clients have a
bursty traffic pattern.

\subsection{Performance Comparison for Continuous Traffic}\label{sec:result-cont}
We evaluate the performance of the comparison schemes in terms of 1)
throughput gain, 2) fairness, and 3) overhead.

\begin{figure}[t!]
\centering
{
\footnotesize
\begin{tabular}{c}
\epsfig{file=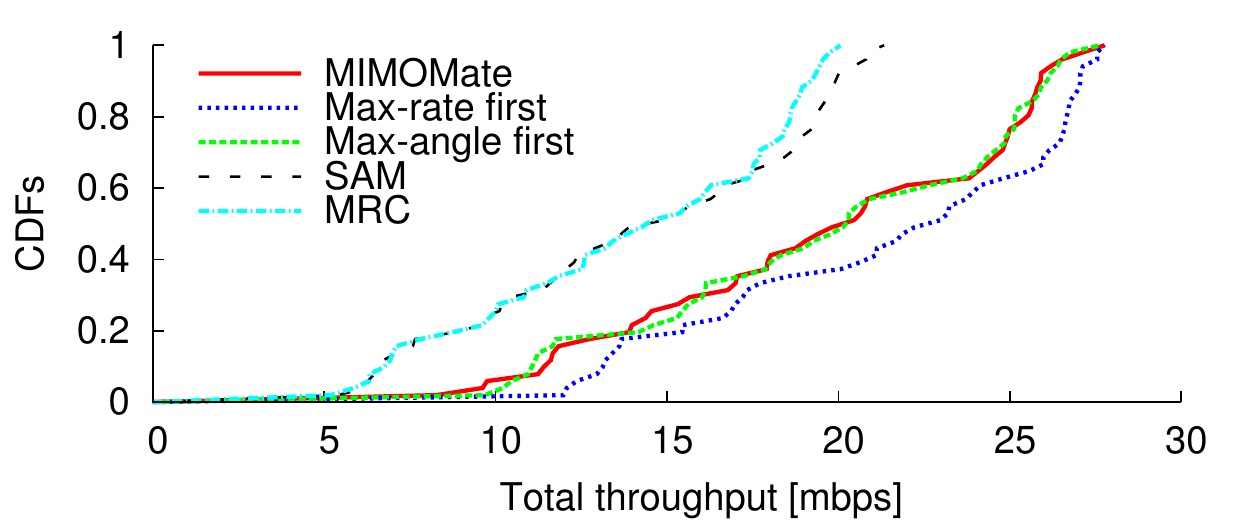, width=3.2in} \\
(a) Total throughput in the 2-antenna AP scenario  \\
\epsfig{file=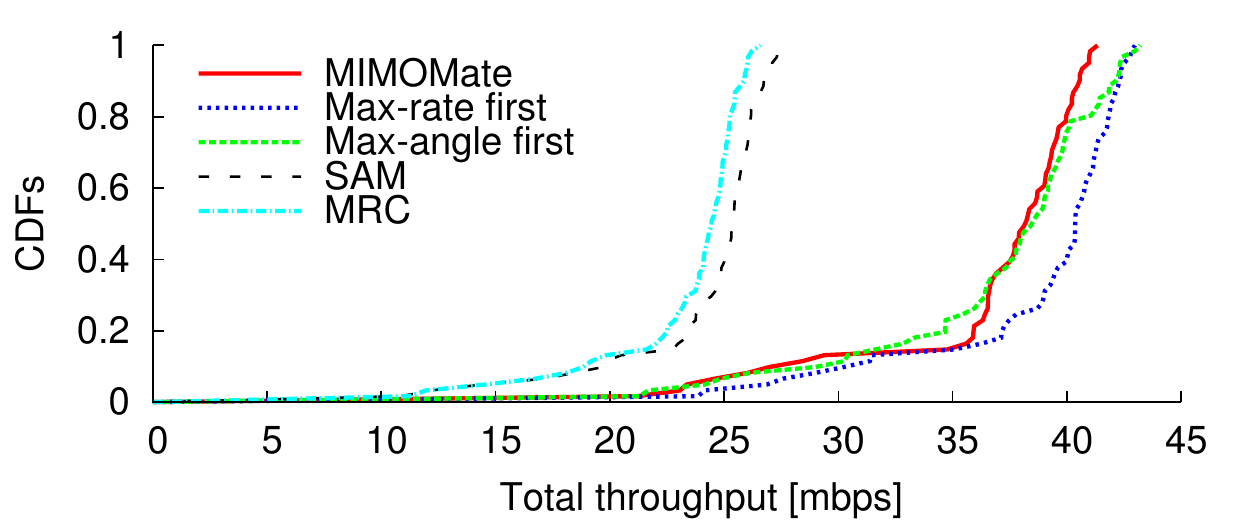, width=3.2in}\\
(b) Total throughput in the 3-antenna AP scenario
\end{tabular}
}
\vspace{-6pt}
\caption{Throughput for continuous traffic}
\label{fig:tput-cont}
\end{figure}

\noindent {\bf Throughput gain:}
We first check the throughput gain delivered by \name\ when users have
a continuous traffic pattern, i.e., always have packets to send.
Hence, in \name, the scheduled {\em MIMO-Mates} can always transmit
concurrently if their lead wins the first contention. We repeat the
experiment with random assignment of client locations in our testbed.

Figs.~\ref{fig:tput-cont}(a) and~\ref{fig:tput-cont}(b) plot the CDFs
of the total throughput in 2-antenna and 3-antenna scenarios,
respectively. The figures show that traditional contention-based
protocols, i.e., SAM and MRC, assign each user an equal probability to
win the contention, without considering the channel characteristics,
and produce a low throughput. MRC requires additional RTS-CTS
overhead, and hence performs a little bit worse than SAM. Compared to
SAM (MRC), the average throughput gain from enabling concurrent
transmissions with \name's user selection is about 42\% (45\%) and
52\% (57\%) in 2- and 3-antenna AP scenarios, respectively.  The gain
mainly comes from two design principles in \name: 1) minimizing SNR
reduction due to MIMO decoding, and 2) reducing the channel time
wasted for contending for concurrent transmissions. Note that the gain
in the 3-antenna AP scenario is higher than that in the 2-antenna AP
scenario. It implies that user matching plays an important role to
deliver the MU-MIMO gain especially when the number of concurrent
transmissions supported by the system increases.  The figures also
show that {\em max-angle first} and {\em max-throughput first} produce
a throughput comparable to (or even slightly higher than) our \name\
because they greedily select the users with the best channel
characteristics or with the highest throughput to join the concurrent
transmissions. In addition, similar to \name, they also require only
one contending process. We will show later that these two schemes
however result in unfair resource sharing.

\begin{figure}[t!]
\centering
{
\footnotesize
\begin{tabular}{c}
\epsfig{file=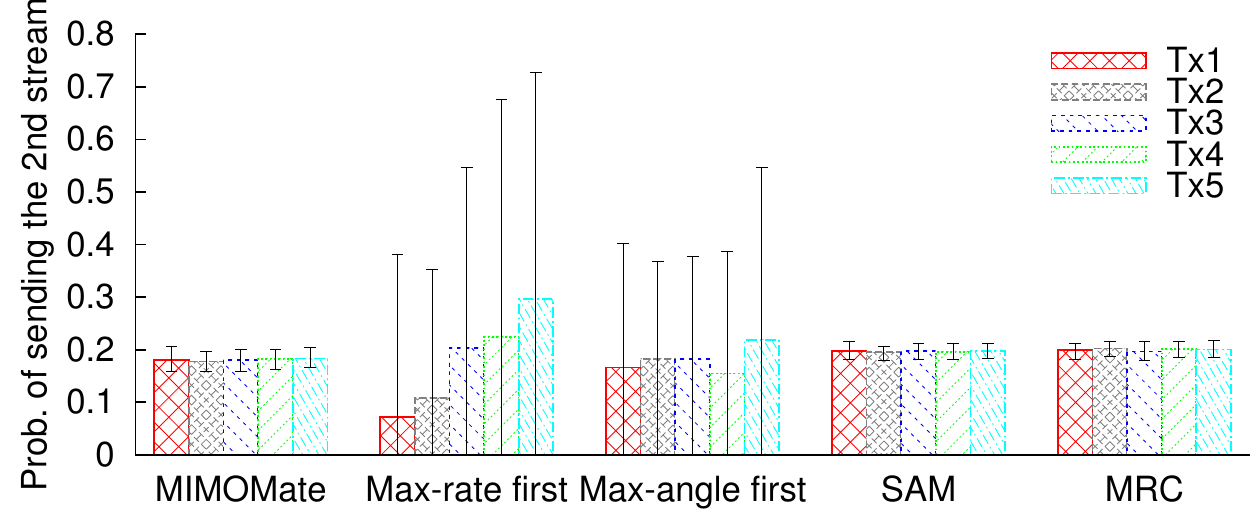, width=3in} \\
(a) Fairness of the second stream  \\
\epsfig{file=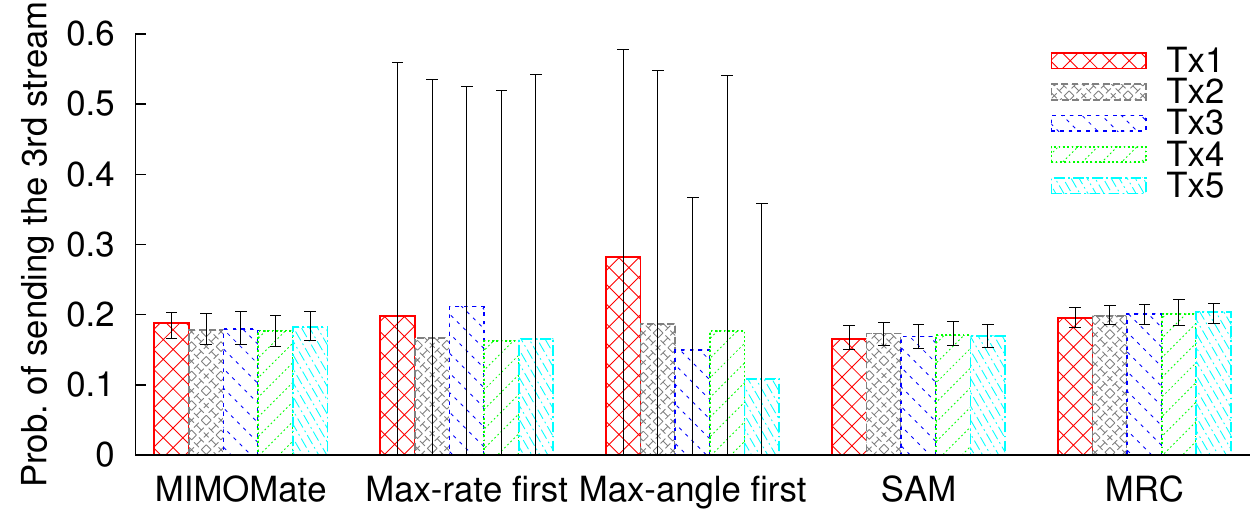, width=3in}\\
(b) Fairness of the third stream
\end{tabular}
}
\vspace{-6pt}
\caption{Fairness comparison}
\label{fig:fairness}
\end{figure}

\noindent {\bf Fairness:}
We next examine fairness of sharing concurrent transmission
opportunities among clients in a 3-antenna AP scenario.  We plot in
Figs.~\ref{fig:fairness}(a) and~\ref{fig:fairness}(b) the number of
the second transmission opportunities and the third transmission
opportunities obtained by each client over the total number of
transmissions, which is the metric used to evaluate fairness in our
experiments. The figures show that both the contention-based schemes,
i.e., SAM and MRC, and our \name\ enable all clients to get almost an
equal probability to transmit the second stream and the third stream,
respectively.  This implies that our matching algorithm enables users
to achieve the same level of fairness as if they use a fair contention
mechanism. The probability of sending the third stream in SAM is
however slightly lower than that in \name\ and MRC.  This is because,
if the transmission time of the first stream is too short due to a
high data rate, then there might be no enough time for SAM to hold the
third stream and its contention. On the other hand, in {\em
max-throughput first} and {\em max-angle first}, users cannot have a
fair opportunity to transmit concurrently because these two schemes
always favor certain users to achieve a high throughput.  Based on the
results in Figs.~\ref{fig:tput-cont} and \ref{fig:fairness}, we
conclude that \name\ achieves a throughput comparable to the greedy
algorithms, while providing users a fairness level similar to the
contention mechanism. 

\begin{figure}[t!]
\centering
\epsfig{file=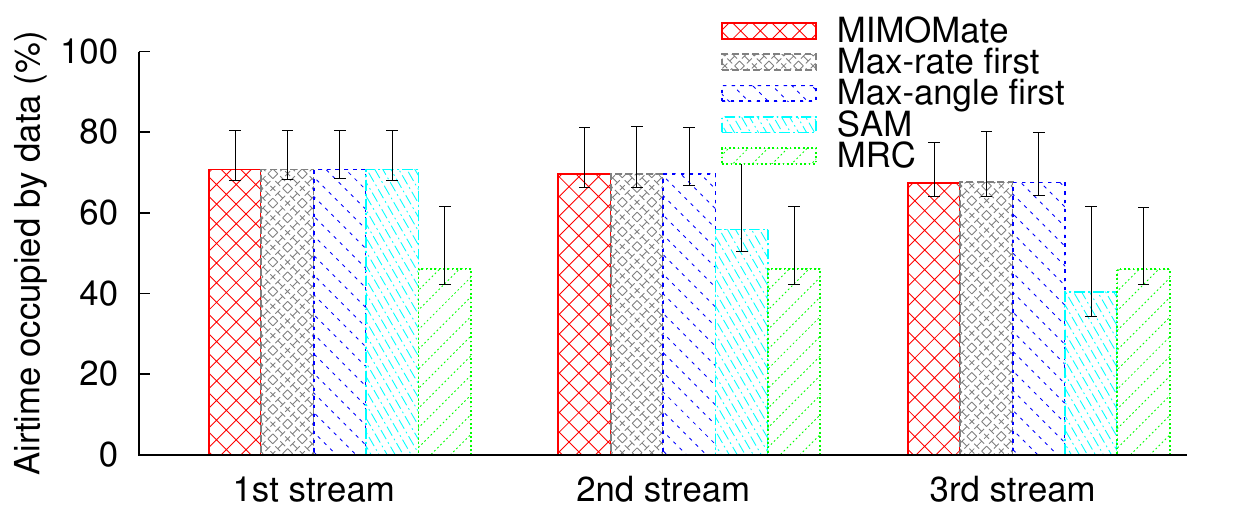, width=3in}
\vspace{-12pt}
\caption{Airtime occupied by the data streams}
\label{fig:overhead}
\end{figure}

\noindent {\bf Overhead:}
We now compare the overhead of different systems. Recall that we do
not implement ACK and contention in USRPs. Thus, we offline compute
the channel time occupied by the protocol overhead. To do so, we feed
the throughput outputted by the experiments (i.e., the rate of the
data frames without considering the 802.11 overhead) to the offline
computation, and add the overhead of each protocol, including
contention, interframe timing (SIFS/DIFS), the PLCP/MAC headers, and
the ACK, to each packet transmission.  Fig.~\ref{fig:overhead} plots
the percentage of airtime occupied by the data frames, which is
computed by the ratio of airtime for data transmission to the overall
airtime occupied by a packet (i.e., including the overhead).  The
figure shows that, since clients in our \name\ and SAM use 802.11's
contention to compete for sending the first stream, their overhead for
the first stream is the same with that of the conventional 802.11. By
eliminating the contending process for the second stream, \name\ and
the greedy algorithms can utilize about 63\% of airtime to transmit
the second streams.  In contrast, both SAM and MRC require multiple
rounds of contention, which significantly offset the MIMO gain when
the number of antennas supported by the AP keeps increasing.  SAM
allows each client to transmit immediately once it wins the
contention. Hence, the airtime of data sent in a higher dimension
becomes shorter and shorter. For MRC, all concurrent clients start
their transmissions after receiving the CTS, their available airtime
is hence shorter yet the same.  However, since multiple clients might
send RTS concurrently in MRC, the number of contention rounds required
in MRC might be fewer than that in SAM. This explains why the airtime
of the third stream in MRC is longer than that in SAM.

\begin{figure}[t!]
\centering
{
\footnotesize
\begin{tabular}{c}
\epsfig{file=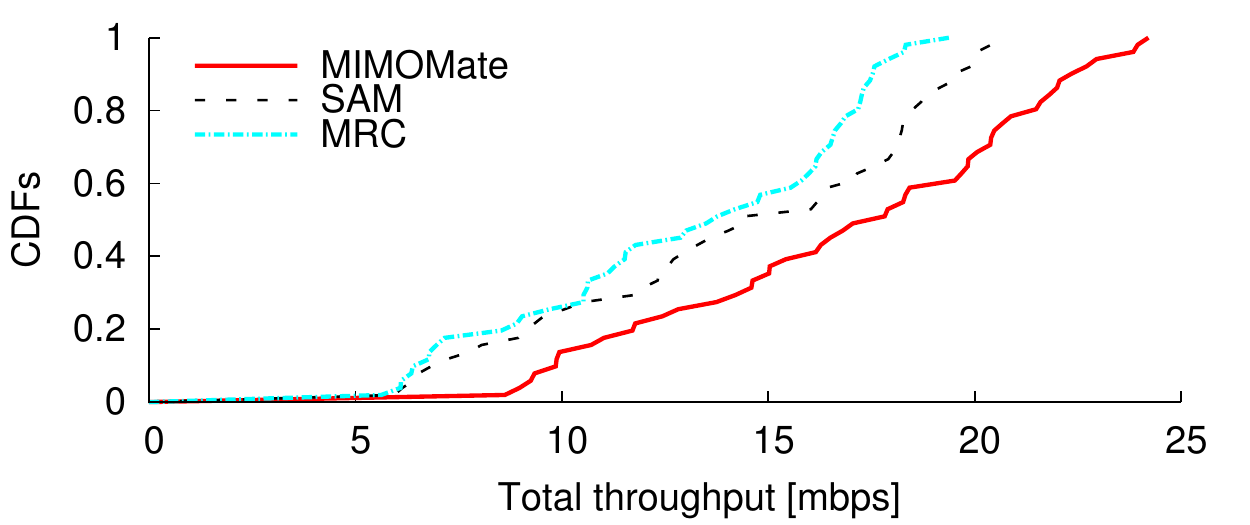, width=3in} \\
(a) Total throughput in the 2-antenna AP scenario  \\
\epsfig{file=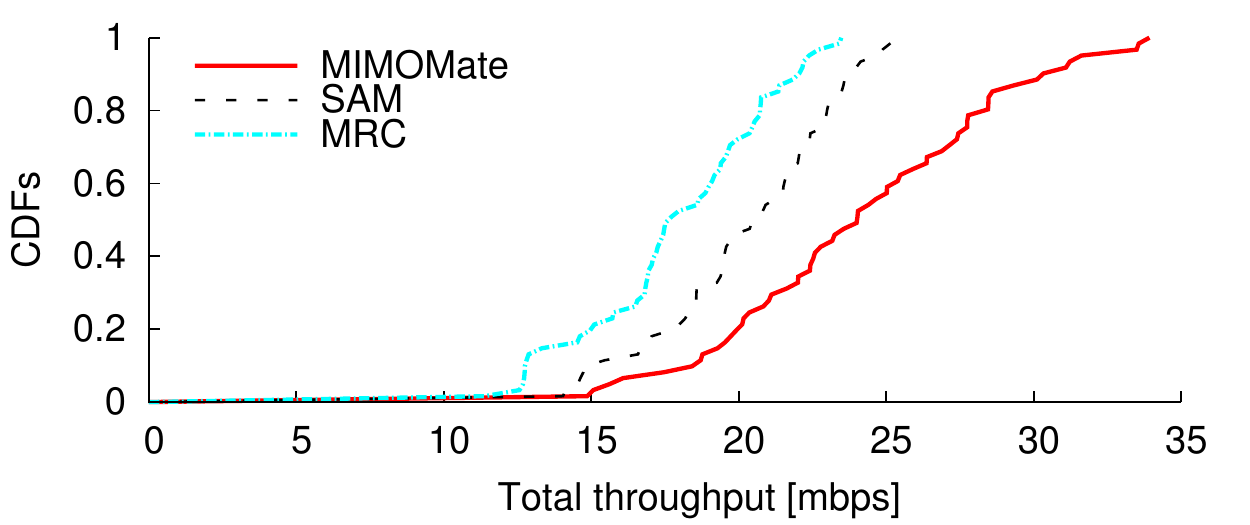, width=3in}\\
(b) Total throughput in the 3-antenna AP scenario
\end{tabular}
}
\vspace{-6pt}
\caption{Throughput for bursty traffic}
\label{fig:tput-bursty}
\end{figure}

\subsection{Throughput Gain for Bursty Traffic}
We next evaluate the performance of \name\ when clients have a bursty
traffic pattern. The packet traces are generated using the following
model. Each user transmits several files to the AP, and the size of
each file is randomly selected from 500 to 550 KB. The arrival of file
transmission follows a Poisson process with an arrival rate $\lambda =
2$ files per second. Thus, when any of the scheduled {\em MIMO-Mates}
does not have traffic to send, other clients contend for transmitting
concurrently using the contention window computed based on
Algorithm~\ref{algo:contention}.  Fig.~\ref{fig:tput-bursty} plots the
CDFs of the total throughput. The figure shows that, compared to SAM
(MRC), the average throughput gain achieved by combining \name\ with
angle-based contention is about 22\% (33\%) and 19\% (36\%) for 2- and
3-antenna scenarios, respectively.  The throughput gain in this case
is lower than that in the continuous traffic scenario because
angle-based contention introduces additional contention overhead.
Also, for some periods, there are only a few clients with traffic
demand and hence concurrent transmission opportunities cannot be fully
utilized. The gap between SAM and MRC for bursty traffic is larger
than that for continuous traffic. This is because, in MRC, when there
are always at least $N$ contending clients, where $N$ is the degrees
of freedom, the AP might be able to detect $N$ clients within fewer
than $N$ contention rounds, if some clients send RTS at the same time
in one contention round. However, when the number of contending
clients is less than $N$, the AP always needs to wait for the duration
of $N$ rounds of RTS and then responds the CTS.

\begin{figure}[t!]
\centering
{
\footnotesize
\begin{tabular}{c}
\epsfig{file=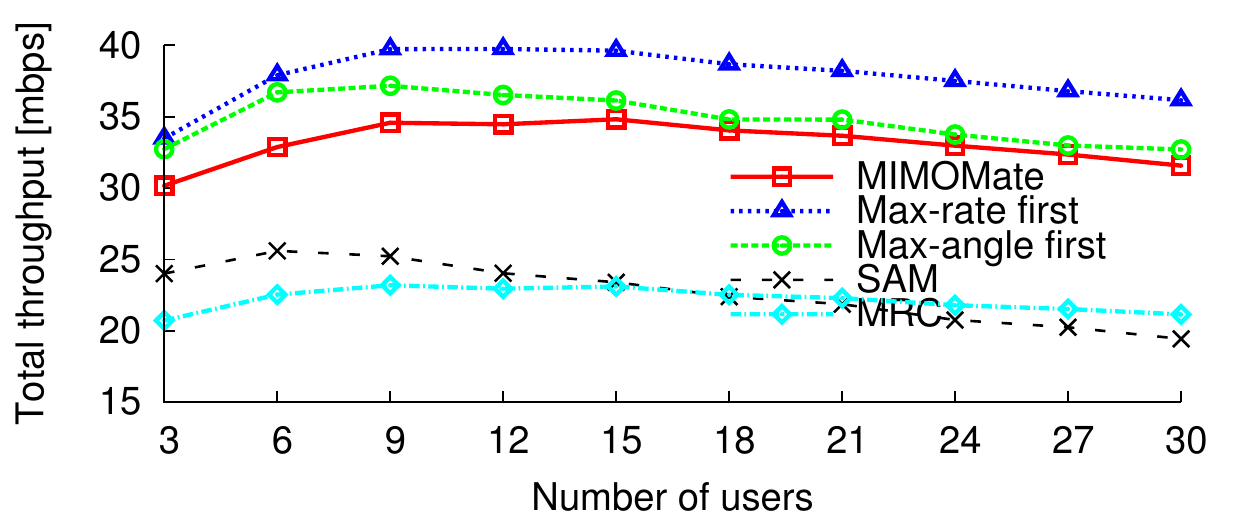, width=3in} \\
(a) Total throughput in the 2-antenna AP scenario  \\
\epsfig{file=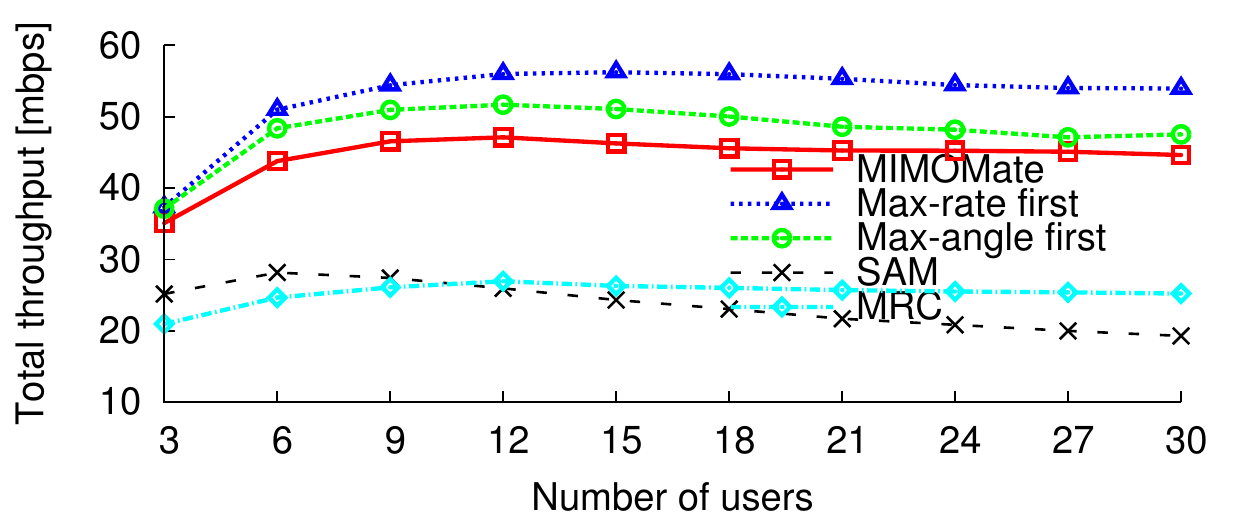, width=3in}\\
(b) Total throughput in the 3-antenna AP scenario
\end{tabular}
}
\caption{Impact of number of users on throughput}
\label{fig:usrnum-tput}
\end{figure}

\section{Simulation Results}\label{sec:simulation}
We further perform simulations to evaluate the performance of \name\
in large-scale scenarios.  The simulations are designed to answer the
following questions.
\begin{itemize}
\item How does \name\ perform in different scales of networks?
\item How does the packet size affect the throughput performance of
comparison protocols?
\item Can legacy 802.11 devices operate normally in the presence of
\name\ nodes?
\end{itemize}

In each simulation, we uniformly randomly distribute the users in a
disk with center the AP and radius 100~m. Furthermore, the antennas on
the AP are collinear with a gap of 0.05~m between two neighboring
antennas.  The channels are generated according to the Rayleigh fading
channel model, and we assume the available transmission bit-rates are
6, 9, 12, 18, 24, 36, 48, and 54~Mb/s, which are identical with
802.11a. The default number of users is set to 15, and the default
packet size is set to 1500 bytes. The detailed settings will be
specified in each simulation.

\begin{figure}[t!]
\centering
{
\footnotesize
\begin{tabular}{c}
\epsfig{file=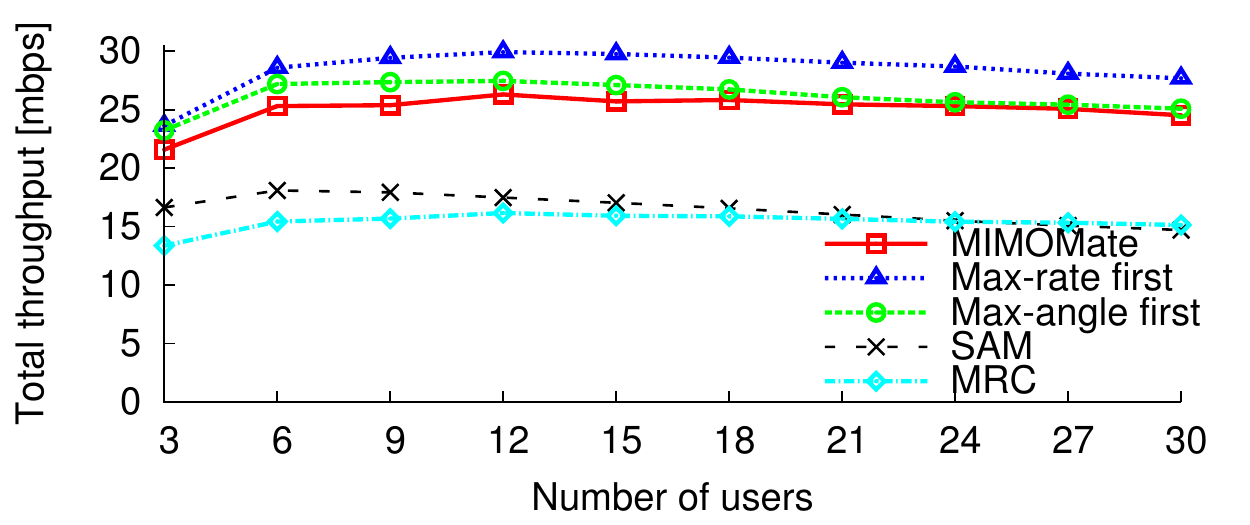, width=3in} \\
(a) Total throughput in the 2-antenna AP scenario  \\
\epsfig{file=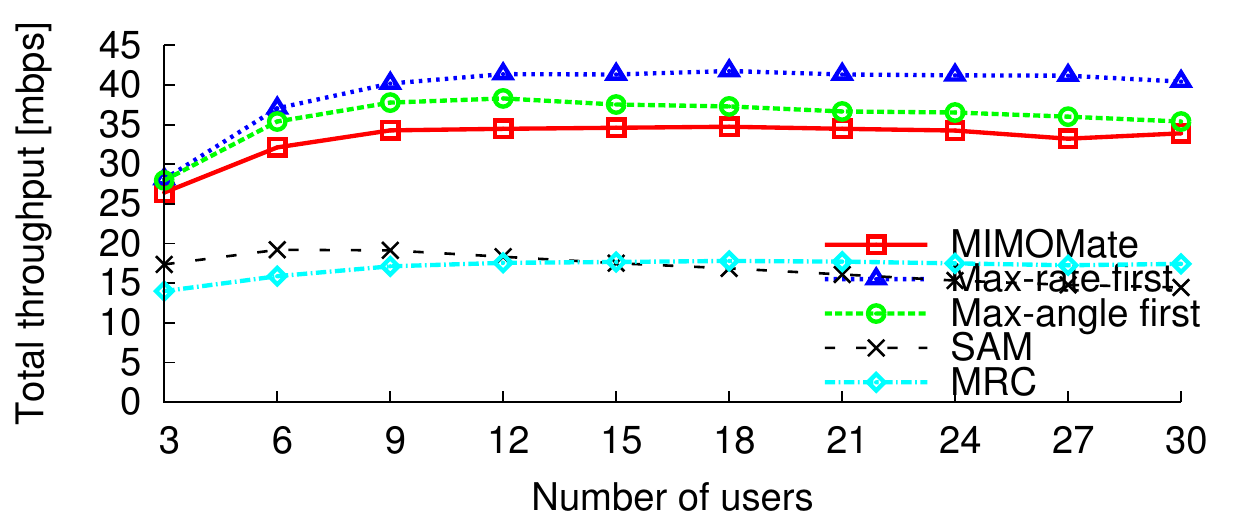, width=3in}\\
(b) Total throughput in the 3-antenna AP scenario
\end{tabular}
}
\caption{Impact of various packet sizes on throughput} \label{fig:pktsize-tput}
\end{figure}

\subsection{Impact of Number of Clients}
We evaluate the performance of \name\ when the number of clients
varies from 3 to 30. Figs.~\ref{fig:usrnum-tput}(a)
and~\ref{fig:usrnum-tput}(b) plot the total throughput for 2-antenna
and 3-antenna AP scenarios, respectively.  The trend of the simulation
results are similar to that of small-scale experimental results.  The
effect of increasing the number of clients on \name, {\em max-angle
first}, and {\em max-throughput first} is relatively small, showing
that \name\ operates well even when the network scales up.  One thing
worth noting is that SAM outperforms MRC when the number of clients is
small, e.g., less than twelve, because its clients do not need to wait
for transmitting concurrently after multiple rounds of contention. The
throughput of SAM however decreases when the number of clients
increases.  The reason is that, although both SAM and MRC require
multiple rounds of contention, contention failure can only occur in
the last round of contention in MRC, yet could happen in any round of
contention in SAM.  The more clients exist, the higher probability
that contention fails is. When contention fails, the number of
concurrent streams would exceed the degrees of freedom, which makes
ZF-SIC decoding fail.  Hence the result.

\subsection{Impact of Dynamic Packet Sizes}
We next evaluate how \name\ performs when the packet size varies
dynamically. In this simulation, we uniformly randomly pick a size
between 200 bytes and 1500 bytes for the first client of each
transmission. The clients joining later end their transmissions at the
same time with the first stream.  Again,
Figs.~\ref{fig:pktsize-tput}(a) and \ref{fig:pktsize-tput}(b) plot the
total throughput for 2-antenna and 3-antenna AP scenarios,
respectively.  Due to a smaller average packet size, i.e., around
$(200+1500)/2$ bytes, and, as a result, a higher proportion of airtime
occupied by the overhead, the throughput in this simulation is less
than those found in Fig.~\ref{fig:usrnum-tput}; However, the advantage
of \name\ does not get affected when the packet size changes.
Observer that, compared to Figs.~\ref{fig:usrnum-tput}(a)
and~\ref{fig:usrnum-tput}(b), the throughput of SAM decreases slower
here.  This is because, in SAM, when the packet size is small, then
there might be no remaining airtime for clients to exploit the
transmission opportunities of high dimensions. In other words, the
number of contentions for SAM in this simulation is less than that in
the previous simulation. As a result, less contention failures occur
here. Therefore, compared with the previous simulation, the effect of
increasing the number of users on SAM is less here.

\begin{figure}[t!]
\centering
{
\footnotesize
\begin{tabular}{c}
\epsfig{file=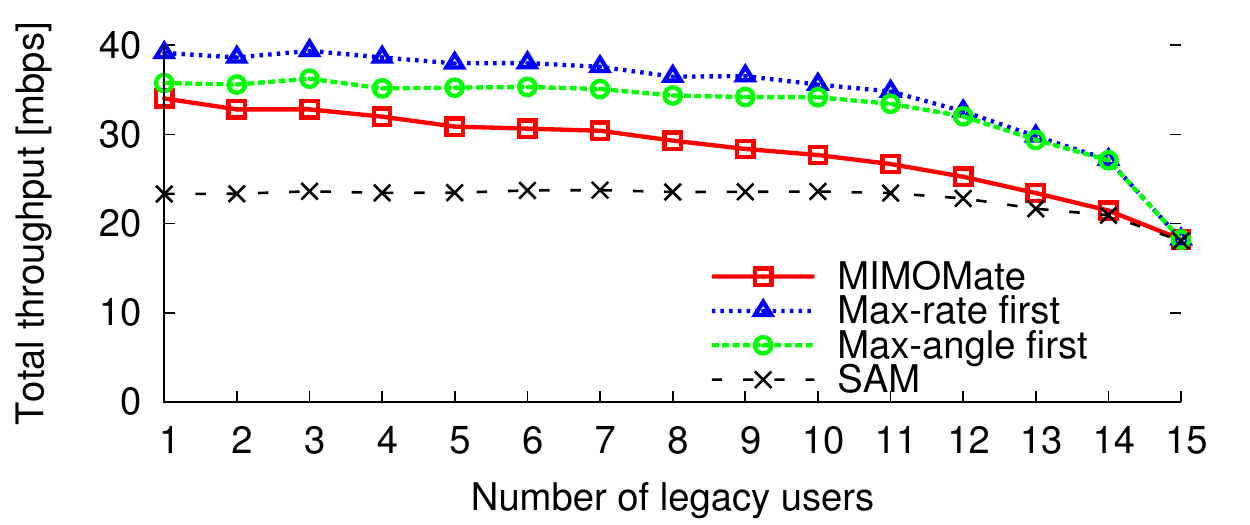, width=3in} \\
(a) Total throughput in the 2-antenna AP scenario  \\
\epsfig{file=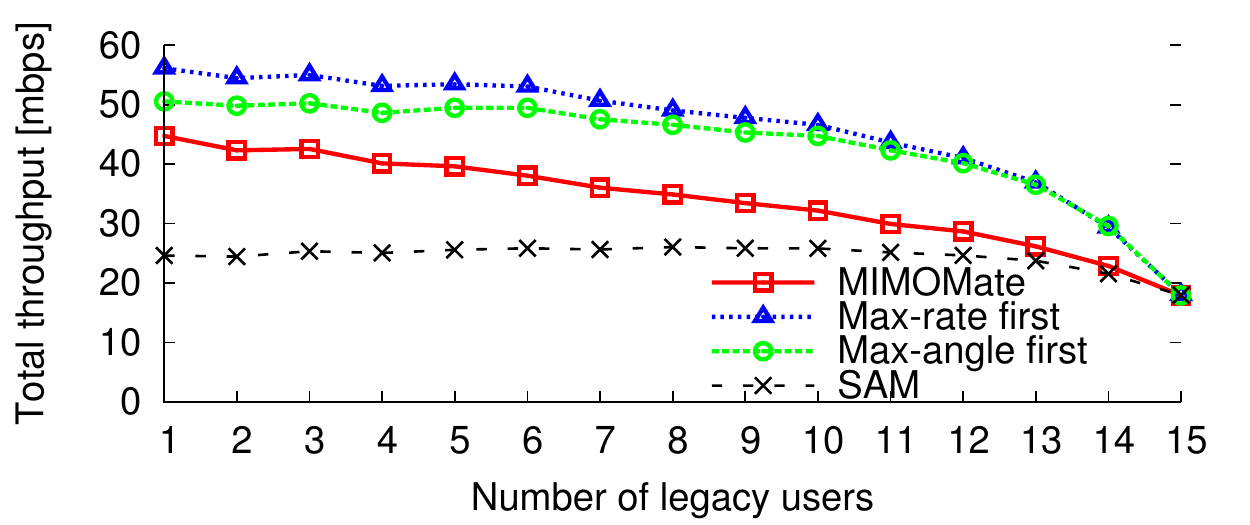, width=3in}\\
(b) Total throughput in the 3-antenna AP scenario
\end{tabular}
}
\caption{Impact of number of legacy nodes on throughput}
\label{fig:legacynum-tput}
\end{figure}

\subsection{Impact of Existence of Legacy Devices}
We finally check the performance of \name\ in the presence of legacy
nodes. Since MRC is not compatible with the traditional 802.11
standard, we exclude it from this simulation.  In this simulation, we
set the total number of clients at 15, and let some of nodes be legacy
802.11 clients. Figs.~\ref{fig:legacynum-tput}(a) and
\ref{fig:legacynum-tput}(b) show the total throughput when the number
of legacy nodes varies from 1 to 15.  The results show that, in
\name\, {\em max-angle first} and {\em max-throughput first}, the
total throughput decreases as the number of legacy nodes increases.
The first reason is that, since legacy nodes can only send the first
stream, i.e., occupying the first dimension, the probability of
picking concurrent clients with a good channel decreases when the
number of non-legacy nodes decreases. Second, when the number of
non-legacy clients is less than the degrees of freedom, the concurrent
transmission opportunities might not be able to be fully utilized, as
a result reducing the throughput significantly.  When all the users
are legacy users, i.e., the number of legacy users is fifteen, all the
methods degenerate to the traditional 802.11 protocol and hence
perform the same.  The performance of SAM however does not change much
with various numbers of legacy users because it simply randomly picks
clients to fully utilize the available degrees of freedom, without
considering channel orthogonality between concurrent clients. On the
contrary, the performance of SAM increases slightly when the number of
non-legacy nodes decreases, because the probability of collisions due
to contention failure decreases when more nodes join contention.

\begin{figure}[t!]
\centering
{
\footnotesize
\begin{tabular}{c}
\epsfig{file=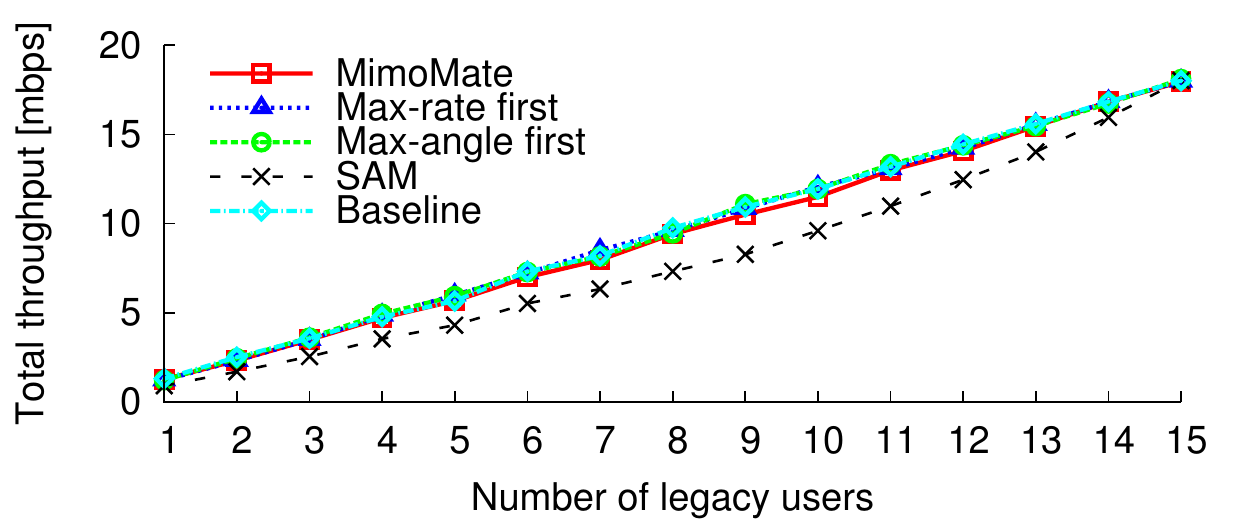, width=3in} \\
(a) Throughput of legacy nodes in the 2-antenna AP scenario  \\
\epsfig{file=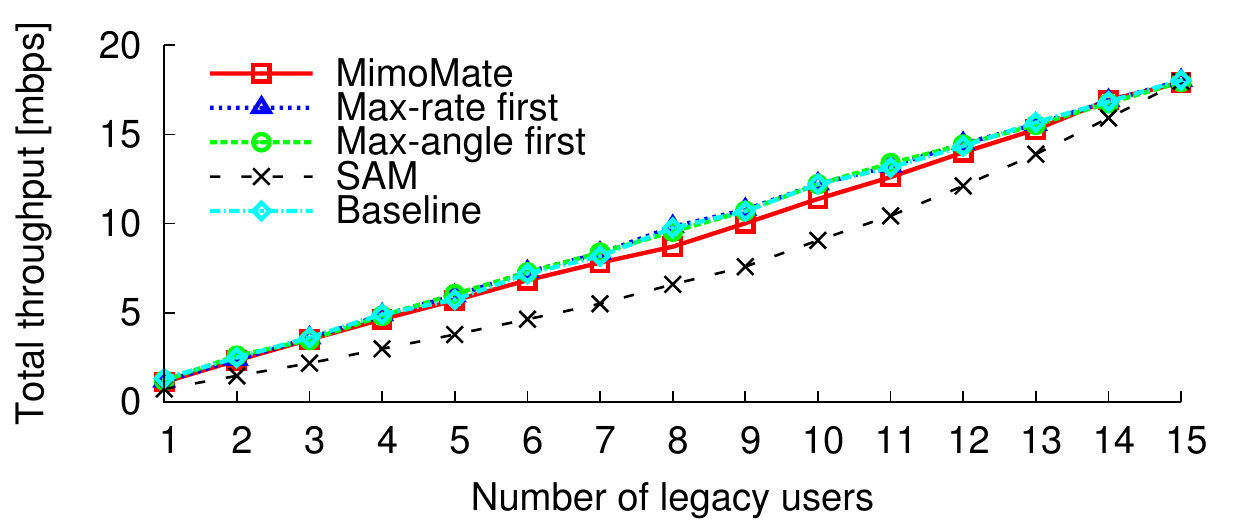, width=3in}\\
(b) Throughput of legacy nodes in the 3-antenna AP scenario
\end{tabular}
}
\caption{Average throughput of legacy devices}
\label{fig:legacytput}
\end{figure}

We further check whether the throughput performance of legacy
	nodes gets affected by our matching design.
	Figs.~\ref{fig:legacytput}(a) and~\ref{fig:legacytput}(b) compare
	the total throughput of legacy nodes in comparison schemes with
	that in the traditional 802.11 protocol. The results show that
	their throughput in SAM is significantly worse than that in
	traditional 802.11. This is because SAM suffers from contention
	failures that might happen in the second and third streams. The
	figure also shows that legacy nodes can coexist with \name\ nodes
	well and achieve a similar performance, as compared to that in
	traditional 802.11.  A small gap between \name\ and traditional
	802.11 is due to occasional contention failures that might happen
	when contention is required for some first winners who are not
assigned any follower.

\section{Conclusion}\label{sec:conclusion}
This paper introduces \name, a user matching protocol that maximally
delivers the gain of a MU-MIMO LAN, while, at the same time, ensuring
all the clients to have a fair opportunity to transmit concurrent
packets. The clients scheduled as the {\em MIMO-Mates} can join
concurrent transmissions one after another with only one contending
process, as a result reducing the MAC overhead significantly. We also
integrate \name\ with an angle-based contention mechanism to best
utilize concurrent transmission opportunities when any of the
scheduled {\em MIMO-Mates} does not have traffic to transmit.  Our
prototype implementation shows that \name\ increases the throughput by
42\% and 52\% over the contention-based protocol for 2- and 3-antenna
AP scenarios, respectively, and also provides fairness for clients.
Our theoretical analysis also proved that \name\ can always achieve a
higher or equal throughput as compared to the traditional contention
scheme in the 2-antenna AP scenario. Analytic performance evaluation
for any general scenarios is considered as our future work.

\appendix
\subsection{Proof of the NP-hardness of the 3-{\em MIMOMate} Problem}
\begin{thrm}
\label{Thrm: MF3_Is_NP}
The 3-{\em MIMOMate} problem is NP-hard.
\end{thrm}
\begin{proof}
We start our proof by relaxing the 3-{\em MIMOMate} problem to a
simpler problem, denoted by the maximum fairness (3MF) problem, which
finds a set $M$ such that Constraints~\ref{Cndtn:
1-3mimomate}--\ref{Cndtn: 4-3mimomate} in problem~\ref{Probdefi:
3mimomate} are satisfied.  Any optimal solution of the 3MF problem is
hence a feasible solution of our problem.  We then proceed the proof
by showing that even the decision version of the relaxed 3MF problem
is NP-hard.

Let 3MF$(V, r, k)$ denote an instance of the decision version of the
3MF problem, which seeks a matching $M \subseteq V \times V$ with
$\vert M \vert \ge k$. We prove its NP-hardness by polynomial-time
reduction from the 3-dimensional matching problem, which is NP-hard.
Let 3DM$(X,Y,Z,T,k)$ denote the decision version of the 3-dimensional
matching problem, which finds a matching $M \subseteq T$, where $T
{\subseteq} X {\times} Y {\times} Z$, such that $u_1{\neq}u_2,
v_1{\neq}v_2$, and $w_1{\neq}w_2$ for any  $(u_1, v_1, w_1), (u_2,
v_2, w_2) \in M$ and $\vert M \vert \ge k$.  For every instance
$3DM(X,Y,Z,T,k)$, we can construct an instance 3MF$(V,r,k)$ in
polynomial time as follows: 

\begin{enumerate} 
\item $V = \{u_i \vert 1 \le i \le \vert X \vert + \vert Y
	\vert + \vert Z \vert\}$.  
\item Set $r^{(u_i,u_j,u_k)}_{u_j} + r^{(u_i,u_j,u_k)}_{u_k}=1,
	\forall (u_i, u_j, u_k)$, if  $(x_i, y_{j-|X|}, z_{k-|X|-|Y|}) \in
	T$; otherwise, set it to 0.
\end{enumerate} 

We show that 3DM$(X,Y,Z,T,k)$ has a feasible solution $M$ if, and only
if, 3MF$(V, r, k)$ has a feasible solution $M'$.  For the ``only if''
direction, for all the elements $(x_i, y_j, z_k)$ in $M$, we add
$(u_i, u_{j+|X|}, u_{k+ |X| + |Y|})$ to the solution of 3MF$(V,r,k)$,
$M'$.  This solution is feasible because, by the construction of the
instance, $r^{(u_i, u_{j+|X|}, u_{k+|X|+|Y|})}_{u_{j+|X|}}+r^{(u_i,
u_{j+|X|}, u_{k+|X|+|Y|})}_{u_{k+|X|+|Y|}} = 1$, and thus Constraints
1 and 3 of the 3MF problem hold.  Constraint 2 is also satisfied,
i.e., $u_i, u_{j+|X|}$ and $u_{k+|X|+|Y|}$ are distinct clients,
because $i \le \vert X \vert < j{+}\vert X \vert \le \vert X
\vert{+}\vert Y \vert < k{+}\vert{X}\vert{+}\vert Y \vert$. Finally,
since different elements $(x_i, y_j, z_k)$ map to different $(u_i,
u_{j+\vert X \vert}, u_{k+ \vert X \vert + \vert Y\vert})$, then
$\vert M' \vert = \vert M \vert \ge k$.

For the ``if'' direction, for every $(u_i, u_j, u_k) \in M'$, we add
$(x_i, y_{j-\vert X \vert}, z_{k - \vert X \vert - \vert Y \vert})$ to
the solution of 3MF, $M$.  Since $(u_i, u_j, u_k)$ in $M'$, we have
$r^{(u_i, u_j, u_k)}_{u_j}+ r^{(u_i, u_j, u_k)}_{u_k}= 1$. Thus, by
the construction of the instance, $(x_i, y_{j-\vert X \vert}, z_{k -
\vert X \vert- \vert Y \vert}) \in T$.  In addition, by
Constraint~\ref{Cndtn: 3-3mimomate}, the constraint of the 3DM problem
that restricts $u_1 {\neq} u_2, v_1 {\neq} v_2$, and $w_1 {\neq} w_2$,
for any two distinct elements $(u_1, v_1, w_1), (u_2, v_2, w_2) \in
T$, holds.  Finally, again, since $\vert M \vert{=}\vert M' \vert$, $M
\ge k$ holds as well.  Hence, we conclude that the 3FM problem is also
NP-hard.
\end{proof}
\subsection{Proofs of Theorem~\ref{thrm: mimomate_vs_contention} 
and Theorem~\ref{thrm: mimomate_vs_random}}
Since Theorem~\ref{thrm: mimomate_vs_random} implies 
Theorem~\ref{thrm: mimomate_vs_contention}, we only need to show 
Theorem~\ref{thrm: mimomate_vs_random}. Furthermore, it is sufficient 
to show that the average throughput 
of the output of Problem~\ref{Probdefi: 2mimomate}, denoted as $T_M$, 
is greater than or equal to that of the optimal fair probabilistic assignment 
selection (the fair probabilistic assignment selection that achieves the 
highest average throughput), denoted as $T_R$.
For ease of presentation, we give an order to the clients, 
so that the input $V$ in Problem~\ref{Probdefi: 2mimomate}, 
i.e., the set of all clients, is $\{1, 2, ..., |V|\}$. 
Note that since the throughput $r^{(i,j)}_j$ is greater than 0 
for all $i,j \in V, i \neq j$, i.e., ZF-SIC decoding is successful, 
a client who transmits the first stream can get about the same 
throughput no matter who its follower is~\cite{TurboRate}. 
In addition, every client has an equal probability 
of winning the first contention. Therefore, we assume that the throughputs 
contributed by the first stream are the same in both the output of  
Problem~\ref{Probdefi: 2mimomate} and the optimal fair probabilistic assignment. 
Thus, we can ignore the average throughput contributed by the first stream 
in the following proof, since we only need to show $T_M \geq T_R$. 

To derive the average throughput, we introduce a variable $p^{i, j}$ for all 
$i, j \in V, i \neq j$. $p^{i, j}$ represents the probability that client $j$ 
is chosen as client $i$'s follower. Therefore, given $p^{i, j}$s, 
the average throughput can be expressed as follows.
\begin{align*}
  \sum_{i \in V}{\{Pr\{\text{client }i \text{ wins the first contention}\}
                \sum_{j \in V \setminus \{i\}}{p^{i, j}r^{(i, j)}_j}\}}     
\\ = \frac{1}{|V|}\sum_{i \in V}{\sum_{j \in V \setminus \{i\}}
                                 {p^{i, j}r^{(i, j)}_j}}.
\end{align*}
Now, we are ready to derive $T_M$ and $T_R$. 
Denote $p_M^{i,j}$s and $p_R^{i,j}$s as the $p^{i, j}$s used in the output of 
Problem~\ref{Probdefi: 2mimomate} and the optimal fair probabilistic assignment, 
respectively. It is then sufficient to show that
\begin{equation}
\label{eq: goal}
\sum_{i \in V}{\sum_{j \in V \setminus \{i\}}{p_M^{i, j}r^{(i, j)}_j}} \geq
\sum_{i \in V}{\sum_{j \in V \setminus \{i\}}{p_R^{i, j}r^{(i, j)}_j}}.
\end{equation}

The following proof is done by three steps: \textbf{1)} show that $p_M^{i, j}$s
is an optimal solution of the integer programming of the bipartite
maximum weighted matching problem; \textbf{2)} show that any $p_R^{i, j}$s is a
feasible solution of the relaxed integer programming of the bipartite
maximum weighted matching problem, which allows the integer variables
to be relaxed as any real number between $[0, 1]$; \textbf{3)} 
Start from the fact that the achieved objective value
of $p_M^{i, j}$s in the integer programming is greater than or equal
to that of $p_R^{i, j}$s in the relaxed integer programming for the
maximum weight matching problem~\cite{matching_theory} and show that 
Eq.~\eqref{eq: goal} holds.

\vskip 0.05in
\noindent {\bf Step 1:} 
Note that $p_M^{i,j}$s are either 0 or 1, indicating whether to use
client $j$ as client $i$'s follower or not.  First observe that since
$r^{(i, j)}_j$ is greater than $0$ for all $i, j \in V, i \neq j$, the
size of the output of Problem~\ref{Probdefi: 2mimomate}, $|M|$, must
be $|V|$.  In other words, every client has a different follower,
which implies
\begin{equation}
\label{eq: property_of_P1}
\sum_{j \in V \setminus \{i\}}{p_M^{i, j}} = 1, \forall i \in V.
\end{equation}
Then, recall that Problem~\ref{Probdefi: 2mimomate} is actually a
bipartite maximum weighted maximum cardinality matching problem.  The
book~\cite{LEDA} shows that, to solve the bipartite maximum weighted
maximum cardinality matching problem, we can add a sufficiently large
number, $C$, to the weight of each edge, and solve the bipartite
maximum weighted matching problem on the new graph instead. 
We can now use the integer programming of the bipartite maximum 
weighted matching problem to find $p_M^{i,j}$s. 
We first give a detailed construction of the 
bipartite graph. Given the set $V$ and throughputs 
$r^{(i, j)}_j$s in Problem~\ref{Probdefi: 2mimomate}, we construct a bipartite 
graph $G = (V_1 \cup V_2, E)$, where\\ 
$V_1 = \{v_1^1, v_1^2, ..., v_1^{|V|}\}$ (the set of winners of the first contention),\\ 
$V_2 = \{v_2^1, v_2^2, ..., v_2^{|V|}\}$ (the set of followers),\\
$E = \{(v_1^i, v_2^j)| r^{(i,j)}_j > 0\}$, and\\
$w(v_1^i, v_2^j) = r^{(i,j)}_j + C, \forall (v_1^i, v_2^j) \in E$ 
(the increased throughput).\\
The integer programming of the bipartite maximum weighted matching problem on $G$ 
can then be formulated as 
\begin{alignat*}{3}
& \text{maximize}   \quad && \sum_{(v_1^i, v_2^j) \in E}{w(v_1^i, v_2^j)p^{i, j}} 
                                        \quad &        \\
& \text{subject to} \quad && \sum_{j \in V \setminus \{i\}}{p^{i, j}} \leq 1,  \forall i \in V, 
                    \\
& \quad &&\text{(the maximum number of edges incident to }v_1^i\text{ is 1)}\\
&                   \quad && \sum_{i \in V \setminus \{j\}}{p^{i, j}}
					\leq 1, \forall j \in V,
					\\
& \quad &&\text{ (the maximum number of edges incident to }v_2^j\text{ is 1)}\\ 
&                   \quad && p^{i, j} = 0 \text{ or } 1, \forall i,j \in V, i \neq j.
\end{alignat*}
The optimal $p^{i,j}$s of the above integer programming are then $p_M^{i,j}$s.

\vskip 0.05in
\noindent {\bf Step 2:} 
To show that $p_R^{i, j}$s correspond to a feasible solution of 
the relaxed integer programming of the bipartite maximum 
weighted matching problem, we need to derive some properties of $p_R^{i, j}$s.
The fairness constraint requires that 
every client has the same probability to transmit the second stream. 
We hence have the following equation. 
\begin{equation}
\label{eq: property_of_ND_1}
\sum_{i \in V \setminus \{j\}} p_R^{i, j} 
= \sum_{i \in V \setminus \{j'\}} p_R^{i, j'} \leq 1, 
\forall j, j' \in V. 
\end{equation}
The inequality must follows. Otherwise, $\sum_{i \in V}\sum_{j \in V
\setminus \{i\}} p_R^{i, j} = \sum_{j \in V}\sum_{i \in V \setminus
\{j\}} p_R^{i, j} > |V|$, which implies $\sum_{j \in V \setminus
\{i\}} p_R^{i, j} > 1$ for some $i$ and contradicts to the fact that
any fair probabilistic assignment chooses at most one follower at each
time, i.e.,
\begin{equation} 
\label{eq: property_of_ND_2}
\sum_{j \in V \setminus \{i\}} p_R^{i, j} \leq 1, \forall i \in V.
\end{equation} 
Then, by Eq.~\eqref{eq: property_of_ND_1} and 
Eq.~\eqref{eq: property_of_ND_2}, we get that 
$p_R^{i, j}$s correspond to a feasible
solution of the relaxed integer programming.

\vskip 0.05in
\noindent {\bf Step 3:} 
It has been shown in~\cite{matching_theory} that,
for the relaxed integer programming of the bipartite maximum weighted 
matching problem, the objective value achieved
by an optimal integral solution is no less than that achieved by
any feasible solution of the relaxed integer programming. 
Hence, we have $\sum_{(v_1^i, v_2^j) \in E}{w(v_1^i, v_2^j)p_M^{i, j}} \geq 
\sum_{(v_1^i, v_2^j) \in E}{w(v_1^i, v_2^j)p_R^{i, j}}$. 
Therefore, it is sufficient to show that 
\begin{equation} \nonumber
\sum_{(v_1^i, v_2^j) \in E}{w(v_1^i, v_2^j)p_M^{i, j}} \geq 
\sum_{(v_1^i, v_2^j) \in E}{w(v_1^i, v_2^j)p_R^{i, j}} 
\text{ implies Eq.~\eqref{eq: goal}}.
\end{equation}
Before showing this, we must further refine Eq.~\eqref{eq: property_of_ND_2} 
and get the following equation.
\begin{equation} 
\label{eq: property_of_ND_3}
\sum_{j \in V \setminus \{i\}} p_R^{i, j} = 1, \forall i \in V,
\end{equation}
which can be proved by contradiction.
If there exists some $i' \in V$, such that 
$\sum_{j \in V \setminus \{i'\}} p_R^{i', j} < 1$, then we must have 
$\sum_{i \in V \setminus \{j\}} p_R^{i, j} 
= \sum_{i \in V \setminus \{j'\}} p_R^{i, j'} < 1, \forall j, j' \in V$;
otherwise, if $\sum_{i \in V \setminus \{j\}} p_R^{i, j} 
= \sum_{i \in V \setminus \{j'\}} p_R^{i, j'} = 1, \forall j, j' \in V$, 
then $\sum_{j \in V \setminus \{i\}} p_R^{i, j} = 1, \forall i \in V$, 
which contradicts to $\sum_{j \in V \setminus \{i'\}} p_R^{i', j} < 1$. 
Therefore, we can add a small value to 
all $p_R^{i', j}$s, $j \in V \setminus \{i'\}$, such that 
Eq.~\eqref{eq: property_of_ND_1} and Eq.~\eqref{eq: property_of_ND_2} 
still hold, i.e., after addition, 
$p_R^{i, j}$s still correspond to a fair probabilistic assignment. 
Obviously, the average throughput is higher 
after we increase $p_R^{i', j}$s, which contradicts to the fact that 
$p_R^{i, j}$s correspond to an optimal fair probabilistic assignment. 
Hence, Eq.~\eqref{eq: property_of_ND_3} holds.

We are now ready to accomplish the final part of the proof.
Observe that 
\begin{align*}
	&\sum_{(v_1^i, v_2^j) \in E}{w(v_1^i, v_2^j)p_M^{i, j}} 
=\sum_{v_1^i \in V_1}{\sum_{v_2^j \in V_2, j \neq i}
                                 {w(v_1^i, v_2^j)p_M^{i, j}}} \\
=&\sum_{i \in V}{\sum_{j \in V \setminus \{i\}}
{(r^{(i, j)}_j + C)p_M^{i, j}}} \\ 
=&\sum_{i \in V}{\sum_{j \in V \setminus \{i\}}
                                 {r^{(i, j)}_j p_M^{i, j}}} 
+C\sum_{i \in V}{\sum_{j \in V \setminus \{i\}}
                                 {p_M^{i, j}}},                       
\end{align*} 
and 
\begin{equation}\nonumber
 \sum_{(v_1^i, v_2^j) \in E}{w(v_1^i, v_2^j)p_R^{i, j}} 
=\sum_{i \in V}{\sum_{j \in V \setminus \{i\}}
                                 {r^{(i, j)}_j p_R^{i, j}}} 
+C\sum_{i \in V}{\sum_{j \in V \setminus \{i\}}
                                 {p_R^{i, j}}} 
\end{equation} 
by a similar reasoning.
Therefore, it is sufficient to show that 
$\sum_{i \in V}{\sum_{j \in V \setminus \{i\}}{p_M^{i, j}}}
= \sum_{i \in V}{\sum_{j \in V \setminus \{i\}}{p_R^{i, j}}}$.
By Eq.~\eqref{eq: property_of_P1} and Eq.~\eqref{eq: property_of_ND_3}, 
we have $\sum_{i \in V}{\sum_{j \in V \setminus \{i\}}{p_M^{i, j}}}
= \sum_{i \in V}{\sum_{j \in V \setminus \{i\}}{p_R^{i, j}}} = |V|$.
The proof is then completed.
\end{document}